\documentclass[aos,preprint]{imsart}

\RequirePackage[OT1]{fontenc}
\RequirePackage{amsthm,amsmath}
\RequirePackage{natbib}
\RequirePackage[colorlinks,citecolor=blue,urlcolor=blue]{hyperref}

\arxiv{arXiv:0000.0000}

\startlocaldefs
\numberwithin{equation}{section}
\theoremstyle{plain}

\endlocaldefs

\usepackage{amsmath,amssymb,amsfonts}
\usepackage{natbib}
\usepackage{color}
\usepackage{pdfsync}
\usepackage{float}
\usepackage{graphicx}
\usepackage{subfigure}

\usepackage{booktabs}
\usepackage{multirow}
\usepackage{tikz}
\usepackage{tkz-euclide}
\usetikzlibrary{decorations.pathreplacing}
\usetkzobj{all}
\usepackage{xcolor,colortbl}
\usepackage{verbatim}
\usepackage{hyperref}

\newtheorem{theorem}{Theorem}
\newtheorem{corollary}{Corollary}

\newtheorem{lemma}{Lemma}

\newtheorem{algorithm}{Algorithm}

\begin{document}

\begin{frontmatter}
\title{Bayesian Nonparametric Tests via Sliced Inverse Modeling}

\begin{aug}
\author{\fnms{Bo} \snm{Jiang}\ead[label=e1]{bojiang83@gmail.com}},
\author{\fnms{Chao} \snm{Ye}\ead[label=e2]{yechao1009@gmail.com}}
\and
\author{\fnms{Jun S.} \snm{Liu}\ead[label=e3]{jliu@stat.harvard.edu}}

\runauthor{Jiang, Ye and Liu}

\affiliation{Harvard University, Tsinghua University and Harvard University}

\address{Department of Statistics\\
Harvard University\\
1 Oxford Street, Cambridge\\
MA 02138, USA\\
\printead{e1}\\
\phantom{E-mail:\ }\printead*{e3}}

\address{Department of Automation\\
Tsinghua University\\
Beijing 100084, China\\
\printead{e2}}
\end{aug}

\begin{abstract}
We study the problem of independence and conditional independence tests between categorical covariates and  a continuous response variable, which has an immediate application in genetics. Instead of estimating the conditional distribution of the response given values of covariates, we model the conditional distribution of covariates given the discretized response (aka ``slices''). By assigning a prior probability to each possible discretization scheme, we can compute efficiently a Bayes factor (BF)-statistic for the independence (or conditional independence) test using a dynamic programming algorithm. Asymptotic and finite-sample properties such as power and null distribution of the BF statistic are studied, and a stepwise variable selection method based on the BF statistic is further developed.  We compare the BF statistic with some existing classical methods and demonstrate its statistical power through extensive simulation studies. We apply the proposed method to a mouse genetics data set aiming to detect quantitative trait loci (QTLs) and obtain promising results. 
\end{abstract}



\end{frontmatter}

\section{Introduction}

Statistical tools for analyzing data sets with categorical covariates and continuous response have been extensively used in many areas such as genetics, clinical trials, social science, and internet commerce. By grouping individual observations according to combinatoric configurations of covariates, classical regression-based methods are derived from conditional models of response given configurations of covariates. Recent demands for analyzing large-scale, high-dimensional data sets pose new challenges to these traditional methods. For example, in quantitative trait loci (QTL) mapping \citep{lander1994, brem2002, morley2004}, scientists wish to discover genomic loci associated with a continuous quantitative trait such as human height, crop yield, or gene expression level, by sequencing hundreds or thousands of genetic markers (encoded as categorical variables) on a genome-wide scale. Regression-based approaches such as analysis of variance (ANOVA) are sensitive to distributional assumption of quantitative traits, and ineffective in detecting individual  markers with heteroscedastic or other higher-order effects. Recently, \cite{aschard2013} proposed a non-parametric method to test whether the distribution of quantitative traits differs by genotypes of a genetic marker. However, for complex traits such as human height, some important genetic markers may have different effects in combination than individually (i.e., the \emph{epistasis} effect in genetic terminology) and thus need to be considered jointly. The number of possible genotype configurations grows exponentially with the number of genetic markers under consideration and, furthermore, markers located on the same chromosome can be highly correlated. Even with a sample size of several hundreds and a moderate number of genetic markers, it is very likely that some genotype configurations contain very few or even no observations. Traditional non-parametric methods have limited power in such situations.

In this paper, we propose a method for detecting associations based on a model of categorical variables conditional on a discretization of the continuous response. This \emph{inverse} modeling perspective is motivated by the na{\" i}ve Bayes method and the ``Bayesian epistatic association mapping'' (BEAM) model of \cite{zhang2007}. BEAM was developed to detect epistatic interactions in genome-wide case-control association studies. Both methods prefer the response variable to be discrete, and their extensions to cases with a continuous response often relies on an {\it ad hoc} discretization strategy. Recently, \cite{jiang2014b} proposed a non-parametric $K$-sample test from the inverse modeling perspective and developed a dynamic slicing (DS) algorithm to determine the optimal discretization (aka ``slicing'') that maximizes a regularized likelihood. In this paper, we employ a full-Bayesian view on the inverse modeling approach  and further generalize the  framework to testing conditional independence between a continuous response and a categorical variable given a set of (previously selected) categorical variables. Instead of constructing test statistics based on regularized likelihood ratios as in DS of \cite{jiang2014b}, we calculate the Bayes Factor (BF) by marginalizing over all possible slicing schemes under the inverse model. From numerical studies, we observed that the BF approach has a superior power in detecting both unconditional and conditional dependences compared with DS and other non-parametric testing methods. The proposed conditional dependence test is further used to construct a stepwise searching strategy for categorical variable selection and applied to QTL mapping analysis. 

The rest of this paper is organized as follows. In Section $2.1$, we construct a non-parametric test based on the Bayes Factor of a sliced inverse model, which we refer to as the BF statistic, to detect the conditional dependence between a continuous response and a categorical covariate.  An efficient dynamic programming algorithm is developed in Section $2.2$ to compute the BF statistic and its asymptotic and finite-sample properties are studied in Section $2.3$. We investigate the sensitivity of the BF statistic to choices of hyper-parameters and fit an empirical formula that links the value of the BF with type-I error in Section $2.4$ and $2.5$, respectively. A forward stepwise variable selection procedure based on the proposed BF statistic is described in Section $2.6$. In Section $3$, we use simulations to evaluate the powers of different methods for both unconditional and conditional dependence tests, and compare the BF statistic with classic stepwise regression in detecting interaction on synthetic QTL data sets. In Section $4$, we further illustrate the proposed methodology on a mouse QTL data set and demonstrate its advantage over traditional QTL mapping methods. Additional remarks in Section $5$ conclude the paper. Proofs of the theorems and other technical derivations are provided in the appendix and online supplement \citep{jiang2015}.

\section{An inverse model for non-parametric dependence test}

Suppose $Y$ is a continuous response variable, and both $X$ and $Z$ are categorical variables with $|X|$ and $|Z|$ levels, respectively. Assume that we have known that $Y$ is dependent of  $Z$. Note that $Z$ can be a ``super'' variable if there are more than one actual variables having been previously selected, in which case we encode each possible configuration of the selected variables as a level of $Z$. For example, if we have selected $Z_1$ and $Z_2$, both of which have support in $\{0,1\}$, then, we can define $Z = Z_1+2Z_2$. Define $Z \equiv 0$ (and thus $|Z|=1$) if we are  interested in testing the marginal independence between $Y$ and $X$. We consider the following hypothesis testing problem:
\begin{eqnarray}\nonumber
&H_0&: \text{$X$ and $Y$ are conditionally independent given $Z$}\\\nonumber 
\text{ v.s. } &H_1&:  \text{$X$ and $Y$ are not conditionally independent given $Z$}
\end{eqnarray}
Note again that  all the results in this section is directly applicable to testing the unconditional dependence between  $X$ and $Y$ (i.e., by letting $Z\equiv 0$). 

Suppose $\{(x_i, y_i, z_i)\}_{i=1}^{n}$ are independent observations of $(X, Y, Z)$. Without loss of generality, henceforth we assume that observations have been sorted according to the $Y$ values so that  $y_i = y_{(i)}$. We divide the sorted list of observations into slices and define a function $S(y_i)$ taking values in $\{1,2,...,|S|\}$ as the slice membership of $y_i$, where $|S|$ denotes the total number of slices. Under the null hypothesis, the conditional distribution of $X$ given $Z$ does not depend on $Y$ and
\begin{equation}\label{eq:h0}
X \mid \  Y = y, Z = j \sim \text{Multinomial}\left(1,p_j \right),
\end{equation}
where $p_j = \left(p_{j,1},\ldots,p_{j,|X|}\right)$ and $\sum_{k=1}^{|X|} p_{j,k}=1$ for $j = 1,\ldots,|Z|$. Under the alternative hypothesis, the distribution of $X$ conditional on $Z=j$ and $S(Y)=h$ ($1 \leq h \leq |S|$) is given by
\begin{equation}\label{eq:h1}
X \mid \ Z=j, S(Y) = h \sim \text{Multinomial}\left(1, p_j^{(h)} \right),
\end{equation}
where $p_j^{(h)} = \left(p_{j,1}^{(h)},\ldots,p_{j,|X|}^{(h)}\right)$ and $\sum_{k=1}^{|X|}p_{j,k}^{(h)}=1$ for $j = 1,\ldots,|Z|$ and $h = 1,\ldots,|S|$. \cite{jiang2014b} proposed a dynamic slicing (DS) statistic to test the above hypotheses with $Z\equiv 0$ based on a regularized likelihood ratio. The DS statistic can be generalized to test conditional dependence with $|Z|>1$, but the number of parameters in the model increases dramatically as $|Z|$ increases, which may impair the power of the method. The details of the generalized DS statistic, algorithms and theoretical results are provided in the online supplement \citep{jiang2015}. Here, we explore a different testing approach based on the Bayes factor (BF). In Section $3$, we will show that the proposed BF statistic consistently outperforms the DS statistic under a variety of scenarios in simulations.

\subsection{Bayes factor under inverse model}

Under the null model (\ref{eq:h0}) and the alternative model (\ref{eq:h1}), we further assume the following priors on $p_j$ and $p_j^{(h)}$ (whose dimensionalities are $|X|$), respectively:
\begin{equation}\label{eq:dir}
p_j \sim \text{Dirichlet}\left(\frac{\alpha_0}{|X|}, \ldots, \frac{\alpha_0}{|X|}\right),
\end{equation}
and
$$
p_j^{(h)} \sim \text{Dirichlet}\left(\frac{\alpha_0}{|X|}, \ldots, \frac{\alpha_0}{|X|}\right),
$$
where $\alpha_0 > 0$ is a hyper-parameter. First, we randomly draw a discretization of $Y$, $\{S(y_i)\}_{i=1}^n$, and then conditional on this discretization, the distribution of $X$ then depends jointly on $Z$ and on the slice containing $Y$. With a slight abuse of notation, we let $\text{Pr}_{H_1}\left(X \mid S(Y),Z\right)$ denote the shorthand of the probability of observing $\{X_i=x_i\}_{i=1}^n$ under $H_1$ given $\{z_i\}_{i=1}^n$ and the slicing scheme $\{S(y_i)\}_{i=1}^n$. After integrating out $p_j^{(h)}$, we can  write down the probability
$$
\text{Pr}_{H_1}\left(X\mid S(Y),Z\right) = \prod_{j=1}^{|Z|}\prod_{h=1}^{|S|} \left[\frac{\Gamma(\alpha_0)}{\Gamma\left(\alpha_0 + n_j^{(h)}\right)}\prod_{k=1}^{|X|}\frac{\Gamma\left(n_{j,k}^{(h)}+\frac{\alpha_0}{|X|}\right)}{\Gamma\left(\frac{\alpha_0}{|X|}\right)}\right],
$$
where $n_{j,k}^{(h)}$ is the number of observations with $z_i=j$, $x_i=k$ and $S(y_i)=h$, and $n_j^{(h)} = \sum_{k=1}^{|X|}n_{j,k}^{(h)}$ is the number of observations with $z_i=j$ and $S(y_i)=h$. Similarly, since  $X$ is independent of any slicing of $Y$ conditional on $Z$ under $H_0$,  by integrating out $p_j$ we have  
$$
\text{Pr}_{H_0}\left(X\mid Y,Z\right) = \text{Pr}_{H_0}\left(X\mid Z\right) = \prod_{j=1}^{|Z|} \left[ \frac{\Gamma(\alpha_0)}{\Gamma\left(\alpha_0 + n_j\right)}\prod_{k=1}^{|X|}\frac{\Gamma\left(n_{j,k}+\frac{\alpha_0}{|X|}\right)}{\Gamma\left(\frac{\alpha_0}{|X|}\right)}\right],
$$ 
where $n_{j,k}$ is the number of observations with $z_i=j$ and $x_i=k$, and $n_j = \sum_{k=1}^{|X|}n_{j,k}$ is the number of observations with $z_i=j$.    

Given $n$ observations ranked by their response values, we denote the collection of all possible slicing schemes as $\Omega_n(S)$ and the probability for choosing a slicing scheme $S(\cdot)$ from $\Omega_n(S)$ {\it a priori} as $\text{Pr}\left(S(Y)\right)$. For a slicing scheme $S(\cdot)$ with $|S|$ slices, we assume here that
\begin{equation}\label{eq:slice}
\text{Pr}\left(S(Y)\right) = \pi_0^{|S|-1}(1-\pi_0)^{n-|S|}.
\end{equation}
That is, with probability $\pi_0$, a ``slice'' is ``inserted'' independently  between the $i$th and $(i+1)$th ranked observations, for $i=1,\ldots,n-1$. Given $n$ observations, we re-parameterize the prior as  $\pi_0 \equiv  1/(1+n^{\lambda_0})$  so that on the log-odds scale we have: 
\begin{equation}\label{eq:lambda}
\log\left(\pi_0/(1-\pi_0)\right) \equiv -\lambda_0 \log(n).
\end{equation}

Under $H_1$ and the slicing prior in (\ref{eq:slice}), we have
\begin{equation}\label{eq:sum}
\text{Pr}_{H_1}\left(X|Z,Y\right) = \sum_{S(Y) \in \Omega_n(S)} \text{Pr}_{H_1}\left(X|Z,S(Y)\right)\text{Pr}\left(S(Y)\right).
\end{equation}
Finally, the BF statistic for comparing the model under the alternative hypothesis and the null is defined as the Bayes factor for testing $H_1$ against $H_0$:
\begin{equation}\label{eq:bf}
\text{BF}\left(X|Z,Y\right) = \frac{\text{Pr}_{H_1}\left(X|Z,Y\right)}{\text{Pr}_{H_0}\left(X|Z,Y\right)} = \sum_{S(Y) \in \Omega_n(S)} \text{BF}\left(X|Z,S(Y)\right)\text{Pr}\left(S(Y)\right),
\end{equation}
where 
$$
\text{BF}\left(X|Z,S(Y)\right) = \frac{\text{Pr}_{H_1}\left(X|Z,S(Y)\right)}{\text{Pr}_{H_0}\left(X|Z,Y\right)}.
$$
We will describe an efficient algorithm to calculate the BF statistic (\ref{eq:bf}) next.

\subsection{A dynamic programming algorithm}
To avoid a bruteforce  summation over $2^{n-1}$ possible slicing schemes in $\Omega_n(S)$, we use a dynamic programming algorithm to calculate $\text{BF}\left(X|Z,Y\right)$ in (\ref{eq:bf}) as follows:
\begin{algorithm}\label{alg:dyn} \ 
\normalfont 
\begin{itemize}
\item \emph{Step $1$}: Rank observations according to the observed values of $Y$, $\{y_i\}_{i=1}^n$. Slicing is only allowed along the ranked list of observations.
\item \emph{Step $2$:} For $1 \leq s \leq t \leq n$, calculate
$$
\psi_{s,t} = \prod_{j=1}^{|Z|} \left[\frac{\Gamma(\alpha_0)}{\Gamma\left(\alpha_0 + n_j^{(s:t)}\right)}\prod_{k=1}^{|X|}\frac{\Gamma\left(n_{j,k}^{(s:t)}+\frac{\alpha_0} {|X|}\right)}{\Gamma\left(\frac{\alpha_0}{|X|}\right)}\right],
$$
where $n_{j,k}^{(s:t)}$ is the number of observations with $z_i=j$ and $x_i=k$ for $s \leq i \leq t$, and $n_j^{(s:t)} = \sum_{k=1}^{|X|} n_{j,k}^{(s:t)} $ is the number of observations with $z_i=j$ for $s \leq i \leq t$.
\item \emph{Step $3$}: Fill in entries of the table $\left\{w_t\right\}_{t=1}^n$ (define $w_1 \equiv (1-\pi_0)/\pi_0$) recursively for $t=2,\ldots,n$,
$$
w_t = \sum_{s=2}^t w_{s-1}(1-\pi_0)^{t-s+1}\left[\frac{\psi_{1,s-1}\psi_{s,t}}{\psi_{1,t}}\right],
$$
where $w_s$ stores a partial sum of  (\ref{eq:bf}) until the $s$th ranked observation.
Then, 
$$
\mathrm{BF}\left(X|Z,Y\right) = w_n.
$$
\end{itemize}
The computational complexity of the dynamic slicing algorithm is $O(n^2)$. 
\end{algorithm}

\subsection{Asymptotic and finite-sample properties of the BF statistic}
We derive theoretical bounds on type-I errors of the BF statistic under three different sampling schemes (conditional permutation, unconditional sampling under the sharp null, and unconditional sampling under the hierarchical null) and show that, as sample size $n \rightarrow \infty$ and with appropriate choice of hyper-parameters, the BF statistic is almost surely smaller than or equal to $1$ for all three sampling schemes under the null. Moreover, when the alternative hypothesis $H_1$ is true, we prove that the BF statistic goes to infinity at an exponential rate with the increase of sample size. 

Given $n$ observations $\{x_i, y_i,z_i\}_{i=1}^n$, we can calculate the $p$-value of the observed BF statistic by shuffling the observed values of $X$ within each group of observations indexed by $\{i : z_i=j\}$, independently for $j \in \{1,\ldots,|Z|\}$. We call this shuffling scheme  the conditional permutation null.
Let $\text{Pr}_{\text{shuffle}}\left(\text{BF}\left(X|Y,Z\right) > b\right)$ denote the probability of observing a BF value of $b$ or larger using the conditional permutation scheme. We prove the following theorem in Appendix~\ref{app:thm1}.

\begin{theorem}\label{thm:bf1} Assume that the hyper-parameter $0 < \alpha_0 \leq |X|$ and observed sample size $n \geq |X|$. There exists a constant $C_1>0$, which only depends on $|X|$ and $|Z|$, such that
$$
\mathrm{Pr}_{\mathrm{shuffle}}\left(\mathrm{BF}\left(X|Y,Z\right) > b\right) \leq C_1n^{|Z|(|X|-1)} \min \left\{\frac{1}{(\log(b)+1)n^{\lambda_0-3}}, \frac{1}{b}\right\},
$$
for any $b \geq 1$ and $\lambda_0$ as defined in (\ref{eq:lambda}). Thus, 
$$
\mathrm{Pr}_{\mathrm{shuffle}}\left(\mathrm{BF}\left(X|Y,Z\right) > 1\right) \leq \frac{C_1}{n^{\lambda_0-|Z|(|X|-1)-3}}
$$
and $\mathrm{BF}\left(X|Y,Z\right) \leq 1 \text{ a.s.}$ as $n \rightarrow \infty$ for $\lambda_0 > |Z|(|X|-1)+4$.
\end{theorem}

The above definition of type-I error is conditioning on $\{n_{j,k}:j=1,\ldots,|Z|,k=1,\ldots,|X|\}$, \emph{i.e.} the number of observations with $z_j=j$ and $x_i=k$. When the total number of observations $n$ is large, the conditional permutation null can be approximated by the following \emph{sharp} null hypothesis:
$$
H_0^{\mathrm{sharp}}: X | Y = y, Z = j \sim \text{Multinomial}\left(1,p_j \right), j=1,\ldots,|Z|,
$$  
where $p_j$'s are fixed but unknown. The following corollary, whose proof is given in Appendix~\ref{app:cor1}, provides a similar finite-sample bound on unconditional type-I error under the sharp null hypothesis. 
\begin{corollary}\label{cor:bf1}
Assume that the hyper-parameter $0 < \alpha_0 \leq |X|$. When the sharp null hypothesis $H_0^{\mathrm{sharp}}$ is true, there exists a constant $C_2>0$, which only depends on $|X|$ and $|Z|$, such that
$$
\mathrm{Pr}_{\mathrm{sharp}}\left(\mathrm{BF}\left(X|Y,Z\right) > b\right) \leq C_2 n^{|Z|(|X|-1.5+\gamma_0)}\min \left\{\frac{1}{(\log(b)+1)n^{\lambda_0-3}}, \frac{1}{b}\right\}.
$$
for any $b \geq 1$, where $\gamma_0=0.57722\ldots$ is the Euler-Mascheroni constant and $\lambda_0$ is defined in (\ref{eq:lambda}). Thus,
$$
\mathrm{Pr}_{\mathrm{sharp}}\left(\mathrm{BF}\left(X|Y,Z\right) > 1\right) \leq \frac{C_2}{n^{\lambda_0-|Z|(|X|-1.5+\gamma_0)-3}}
$$
and $\mathrm{BF}\left(X|Y,Z\right) \leq 1 \text{ a.s.}$ as $n \rightarrow \infty$ for $\lambda_0 > |Z|(|X|-1.5+\gamma_0)+4$.
\end{corollary}

The sharp null hypothesis assumes that the nuisance frequency parameters $\{p_j\}_{j=1}^{|Z|}$ are fixed but unknown. We may further consider a hierarchical sampling scheme where the frequency parameters are sampled from some unknown distributions with bounded densities. This is especially relevant when we repeat the dependence test on a collection of covariates (\emph{e.g.} genetic markers) with the same number of categories but varying marginal frequencies. Specifically, we consider the hierarchical null hypothesis as follows:
\begin{eqnarray}\nonumber
& H_0^{\mathrm{hierar}}:  & X | Y = y, Z = j \sim \text{Multinomial}\left(1,p_j \right),\\\nonumber 
& & \text{ and } p_j \sim f_j\left(p_j\right), j =1,\ldots,|Z|,
\end{eqnarray}
where there exists $M_0>0$ such that the unknown prior density $f_j\left(p_j\right)\leq M_0$ for $j=1,\ldots,|Z|$. When $H_0^{\mathrm{hierar}}$ is true, the Dirichlet prior (\ref{eq:dir}) with $\alpha_0 \leq |X|$ has a positive probability of overlapping with the \emph{true} distribution of $p_j$. Thus, we can obtain a tighter type-I error bound under the hierarchical null hypothesis according to the following corollary proved in Appendix~\ref{app:cor2}.
\begin{corollary}\label{cor:bf2}
Assume that the hyper-parameter $0 < \alpha_0 \leq |X|$. When the hierarchical null hypothesis $H_0^{\mathrm{hierar}}$ is true, there exists a constant $C_3>0$, which only depends on $|X|$ and $|Z|$, such that
$$
\mathrm{Pr}_{\mathrm{hierar}}\left(\mathrm{BF}\left(X|Y,Z\right) > b\right) \leq C_3\min \left\{\frac{1}{(\log(b)+1)n^{\lambda_0-3}}, \frac{1}{b}\right\},
$$ 
for any $b \geq 1$ and $\lambda_0$ as defined in (\ref{eq:lambda}). Thus,
$$
\mathrm{Pr}_{\mathrm{hierar}}\left(\mathrm{BF}\left(X|Y,Z\right) > 1\right) \leq \frac{C_3}{n^{\lambda_0-3}},
$$
and $\mathrm{BF}\left(X|Y,Z\right) \leq 1 \text{ a.s.}$ as $n \rightarrow \infty$ for $\lambda_0 > 4$ and $\alpha_0 \leq |X|$.
\end{corollary}
\noindent Notably, compared with Theorem~\ref{thm:bf1} and Corollary~\ref{cor:bf1}, the finite-sample bound in Corollary~\ref{cor:bf2} depends on the number of categories $|X|$ and $|Z|$ only through $C_3$, which is a constant with respect to sample size $n$ and cutoff $b$.

Next, we show that under $H_1$, $\mathrm{BF}\left(X|Y,Z\right)$ goes to infinity with an exponential rate proportional to  the sample size and the conditional mutual information between $X$ and $Y$ given $Z$, $\text{MI}\left(X,Y|Z\right)$.
\begin{theorem}\label{thm:bf2}
Assume that hyper-parameters $\alpha_0$ and $\lambda_0$ as defined in (\ref{eq:dir}) and (\ref{eq:lambda}) satisfying $0 < \alpha_0 \leq |X|$, $\lambda_0 \geq 1$ and $\lambda_0 = o(n^{\frac{1}{3}}/\log(n))$. Under the regularity condition in Appendix~\ref{app:thm2},
$$
\mathrm{Pr}\left(\mathrm{BF}\left(X|Y,Z\right) \geq e^{n\left[\mathrm{MI}(X,Y|Z) - \delta(n)\right]}\right) \geq  1-4n^{-\frac{1}{32}\log(n)},
$$
where
$$
\delta(n) = O\left(\frac{(\lambda_0+|Z|(|X|-1.5+\gamma_0)/3)|Z|\log(n)}{n^{1/3}}\right) \rightarrow 0
$$ 
as $n \rightarrow \infty$. Thus, $\mathrm{BF}\left(X|Y,Z\right) \geq e^{n\left[\mathrm{MI}(X,Y|Z) - \epsilon\right]}$ a.s. for any $\epsilon>0$ as $n \rightarrow \infty$. 
\end{theorem}
\noindent Note that the conditional mutual information $\mathrm{MI}(X,Y|Z) > 0$ if and only if $X$ and $Y$ are not conditionally independent given $Z$. Theorem~\ref{thm:bf1} and~\ref{thm:bf2} guarantee the consistency of the BF statistic in testing dependence given any finite threshold.

The requirement of $\lambda_0 \geq 1$ in Theorem~\ref{thm:bf2} is sufficient but not necessary. In Section~\ref{sec:lambda}, through simulation studies, we show that the BF statistic can approach infinity as sample size increases under some $\lambda_0 < 1$. However, when the value of $\lambda_0$ is small enough, the BF statistic will converge to zero as shown in Figure~\ref{fig:lambda}. Intuitively, this phenomenon can be explained by the fact that too much weight is given to configurations with bad slicings (swamped by the ``entropy'' effect). On the other hand, when $\lambda_0$ as defined in (\ref{eq:lambda}) is large relative to the sample size (implying a very small $\pi_0$), the $\delta(n)$ term in Theorem~\ref{thm:bf2} will no longer converge to zero and the BF statistic will not be able to differentiate $H_1$ from $H_0$. For example, one can show that when $\lambda_0=O(n)$, the BF statistic $\mathrm{BF}\left(X|Y,Z\right) \rightarrow 1$ almost surely as $n \rightarrow \infty$. Unlike the DS statistic, which is monotonically increasing as $\lambda_0$ becomes smaller, the relationship between the BF statistic and the hyper-parameter $\lambda_0$ is not monotonic. In the following section, we study the sensitivity of the BF statistic and its type-I error to the choice of $\lambda_0$ based on numerical simulations.

\begin{figure}[ht]
\begin{subfigure}{
\includegraphics[width=0.33\textwidth, angle=-90]{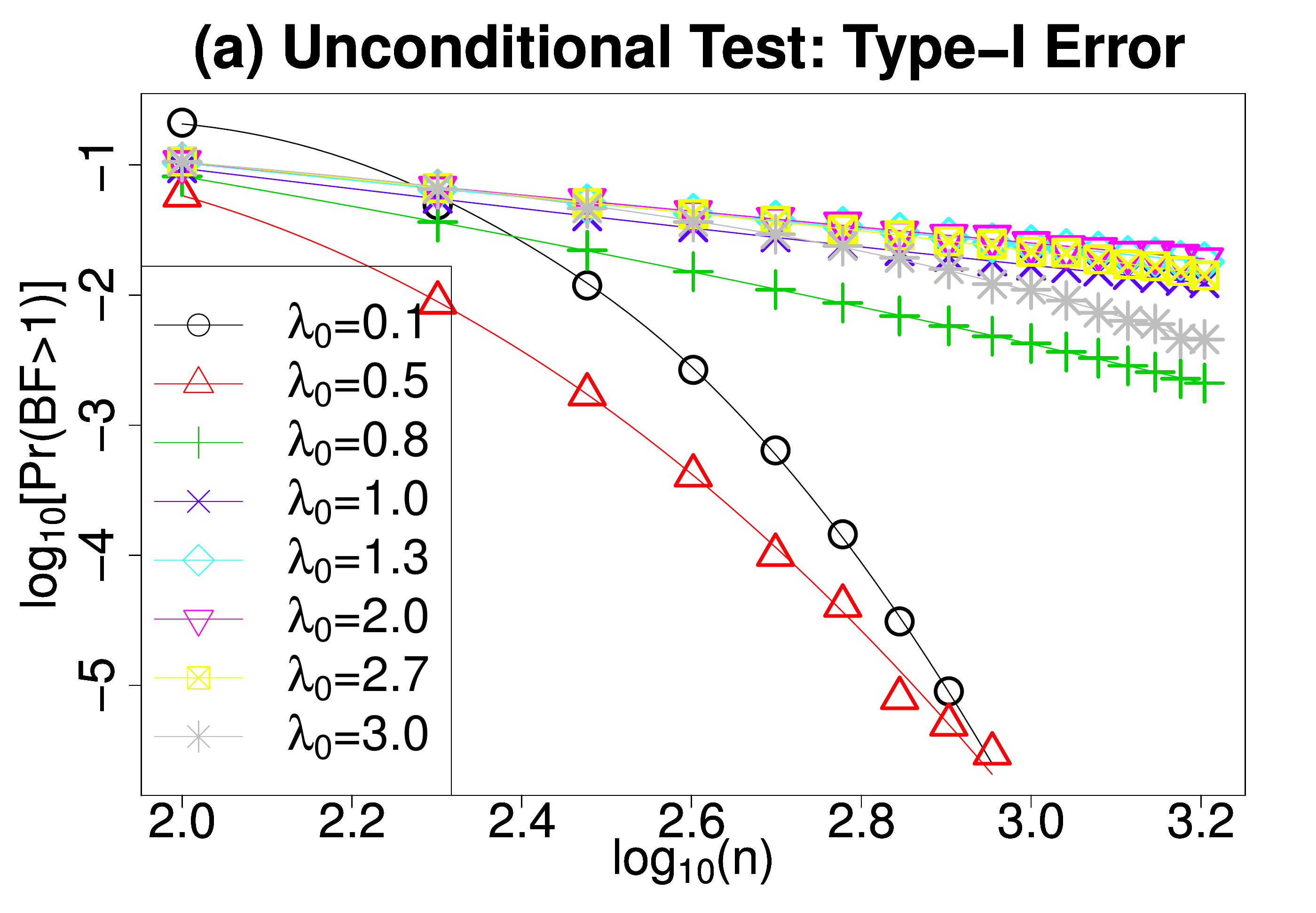}}
\end{subfigure}
\begin{subfigure}{
\includegraphics[width=0.33\textwidth, angle=-90]{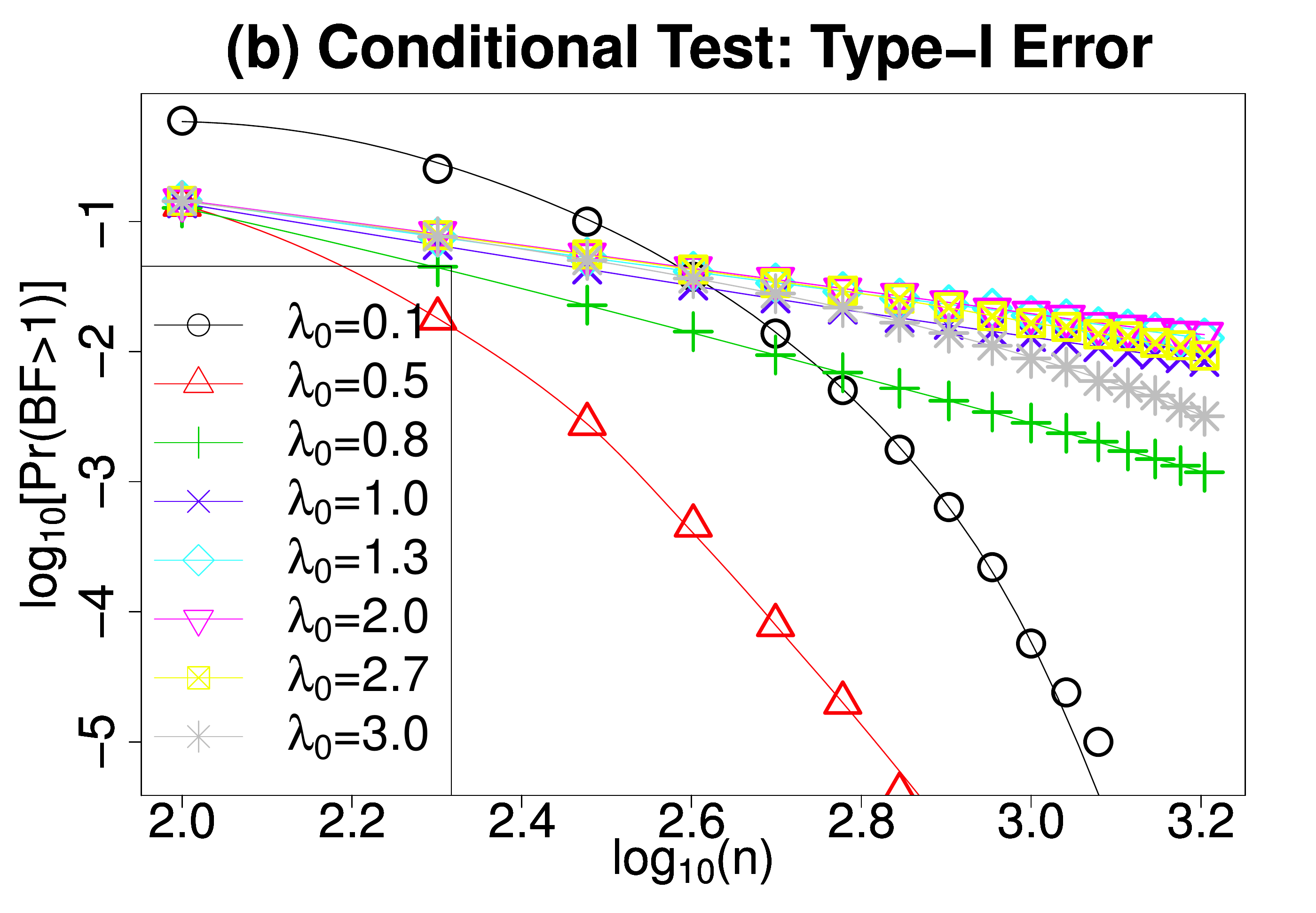}}
\end{subfigure}
\begin{subfigure}{
\includegraphics[width=0.33\textwidth, angle=-90]{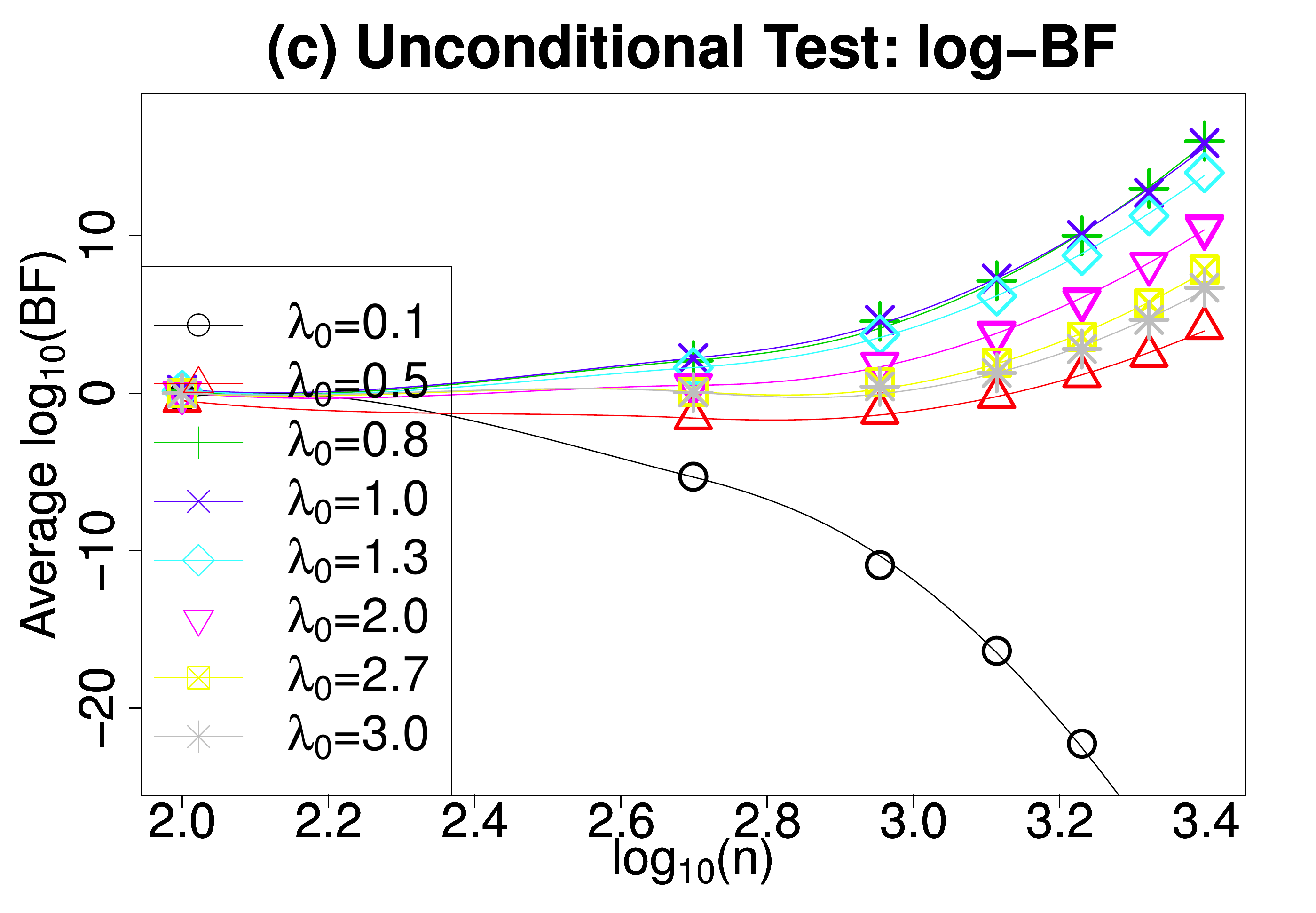}}
\end{subfigure}
\begin{subfigure}{
\includegraphics[width=0.33\textwidth, angle=-90]{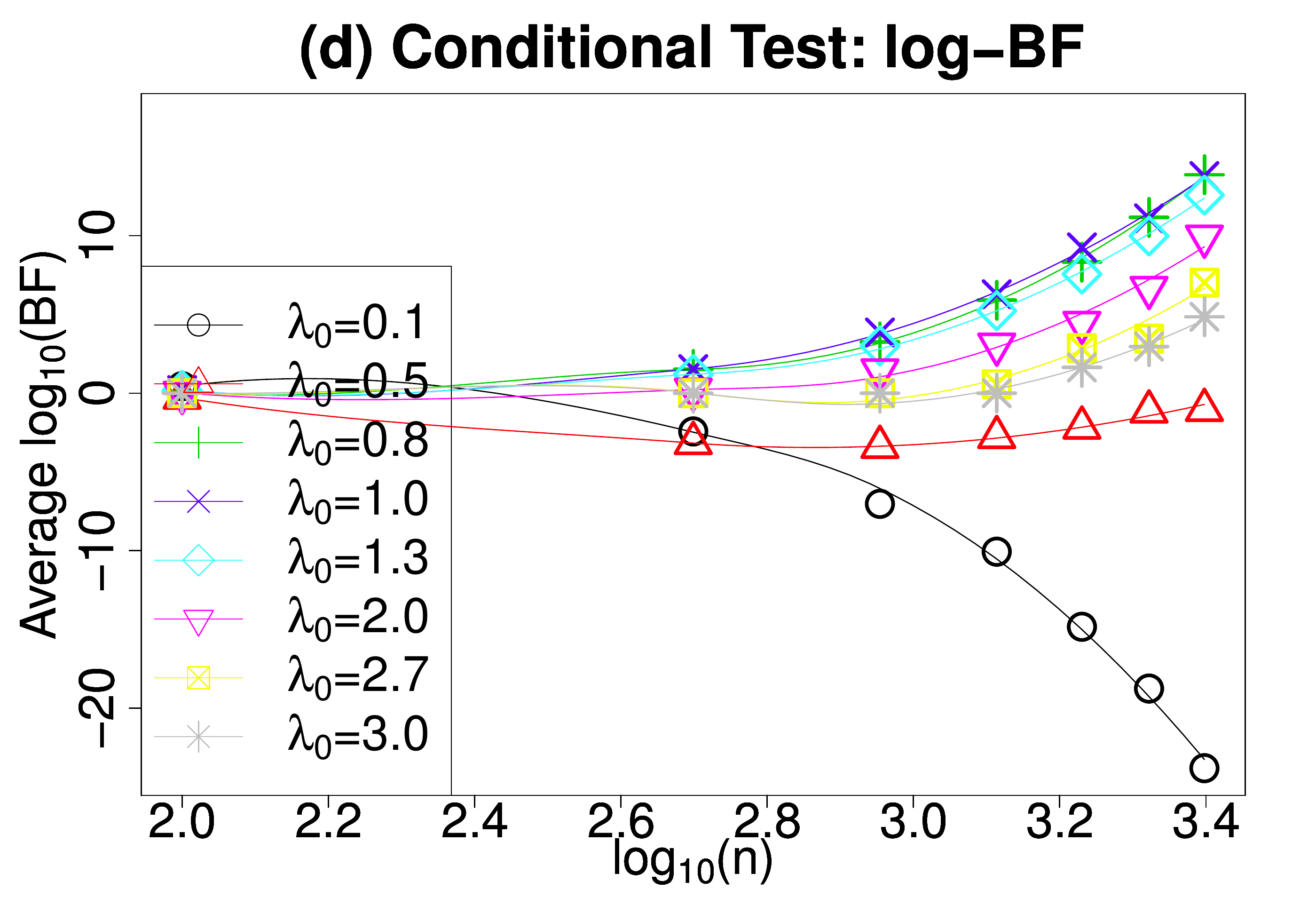}}
\end{subfigure}
\centering
\caption{Sensitivity of type-I error (corresponding to a cutoff of $1$) to different values of $\lambda_0$ for unconditional (a) and conditional (b) BF statistics,
and average logarithm of unconditional (c) and conditional (d) BF statistics given different values of $\lambda_0$ when an alternative hypothesis is true. The lines connecting the points are from the LOESS fit. We fixed $\alpha_0=1$ in this analysis.}
\label{fig:lambda}
\end{figure}

\subsection{Choice of hyper-parameter $\lambda_0$}\label{sec:lambda}

Theorems $1$ and $2$ suggests that we should generally choose hyper-parameters  $\alpha_0$ and $\lambda_0$ as defined in (\ref{eq:dir}) and (\ref{eq:lambda}) such that $\alpha_0 \leq |X|$ and $\lambda_0 \geq 1$. Using numerical simulations, we further study the sensitivity of the BF statistic and its type-I error to the choice of hyper-parameter $\lambda_0$ in both unconditional and conditional dependence tests. For unconditional test, we generate equal number observations with binary indicator $X=0$ and $X=1$, and simulate the continuous response $Y \sim N(\mu X, 1)$ with $\mu = 0.4$. For conditional test, we generate binary covariates $X$ and $Z$ independently, and simulate the response $Y \sim N(\mu X + \mu Z, 1)$ with $\mu = 0.4$. We calculate the average logarithmic value of BF statistics under the alternative hypothesis, as well as type-I error of BF statistic given a cutoff of $1$ under the null hypothesis, which is obtained by shuffling the observed values of $X$ (while retaining the association between $Z$ and $Y$ for conditional test). We use $\alpha_0=1$ in all the simulations in this section and will conduct a sensitivity analysis on $\alpha_0$ in Section $3.1$.

Figure~\ref{fig:lambda} shows the type-I error (given a cutoff of $1$) and average logarithmic value of BF statistic under varying sample size $n$. As we can see, type-I errors under different sample sizes are insensitive to a wide range of $\lambda_0$ from $1$ to $3$. Furthermore, under the alternative hypothesis, values of BF statistics are comparable for choice of $\lambda_0$ between $0.8$ and $1.3$. 
We also observed that the type-I error of the critical region $\{BF>b\}$ is not monotonic to the value of $\lambda_0$. For example, given the same sample size $n$ and a critical value $b=1$, the BF statistic with $\lambda_0=2$ has a larger type-I error than that with $\lambda_0=1$ and $\lambda_0=3$. On the other hand, the BF statistic with $\lambda_0=2$ is on average smaller than the BF statistic with $\lambda_0=1$ when $X$ and $Y$ are (unconditionally or conditionally) dependent. Finally, we note that when $\lambda_0$ is too small (\emph{e.g.}, $0.1$), which results in a relatively large  $\pi_0$ ($\approx n^{-\lambda_0}$) and a large number of expected slices (\emph{i.e.}, $n^{1-\lambda_0}$), the logarithm of BF tends to negative infinity even when $H_1$ is true. Furthermore, it appears that as we vary $\lambda_0$ from 1 to 0, the ``phase transition'' phenomena (\emph{i.e.}, the logarithm of BF diverges to positive infinity versus negative infinity)  occurs at around 0.5. 
Given these observations, unless noted otherwise, we will choose $\lambda_0=1$ and $\alpha_0=1$ (see Section $3.1$ for simulation results using different $\alpha_0$'s) for the following studies.

\begin{figure}[h]
\begin{subfigure}{
\includegraphics[width=0.33\textwidth, angle=-90]{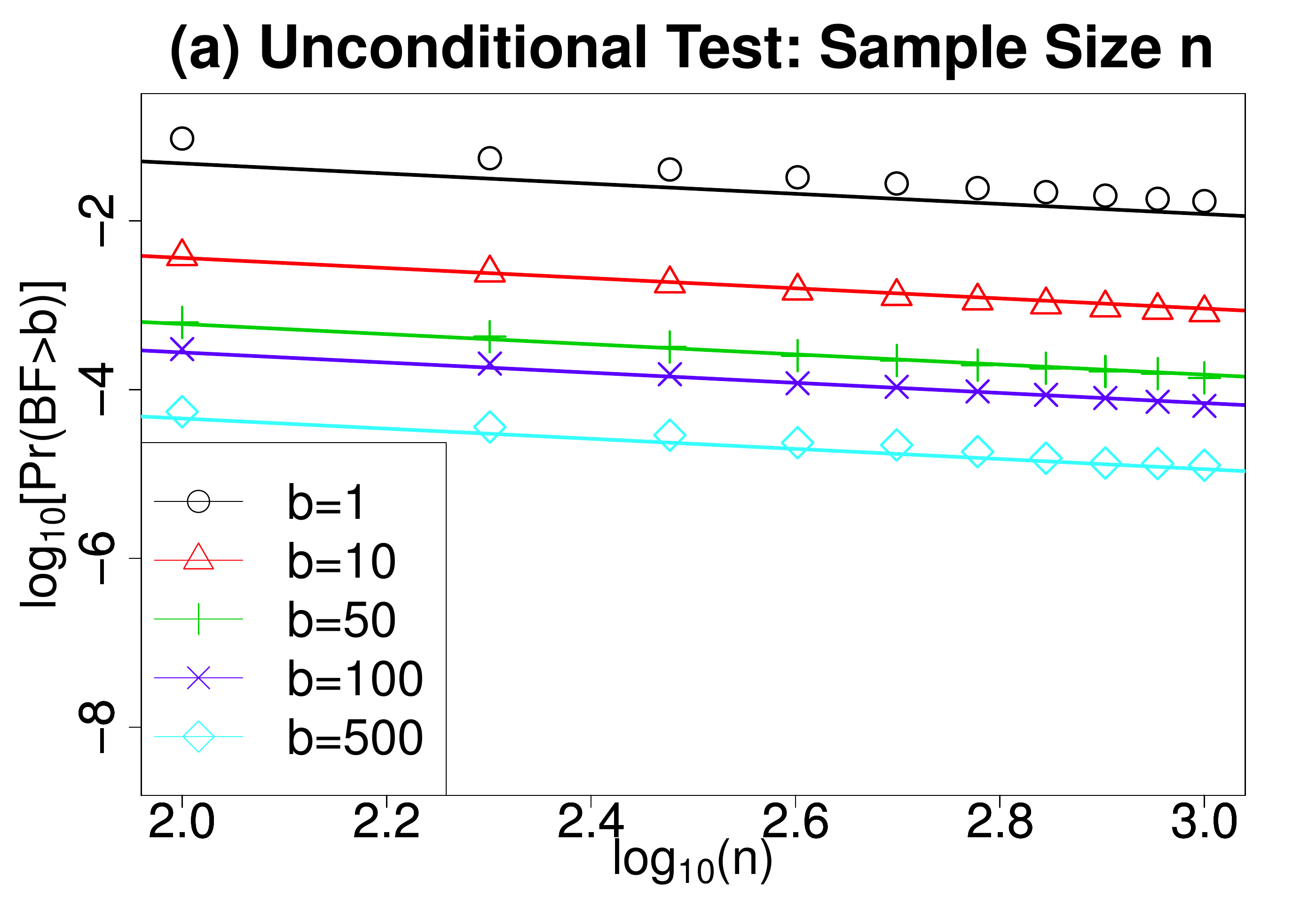}}
\end{subfigure}
\begin{subfigure}{
\includegraphics[width=0.33\textwidth, angle=-90]{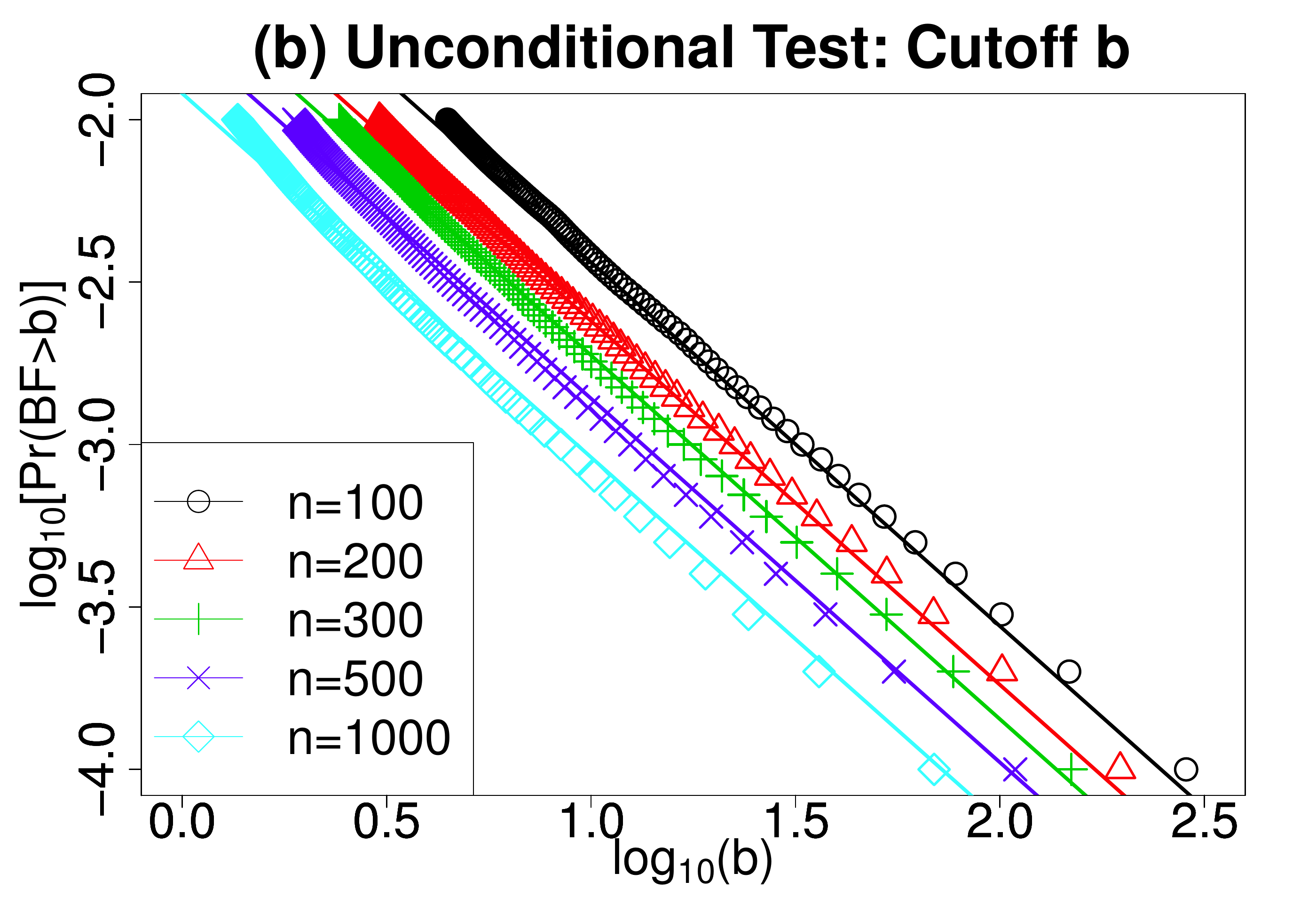}}
\end{subfigure}
\begin{subfigure}{
\includegraphics[width=0.33\textwidth, angle=-90]{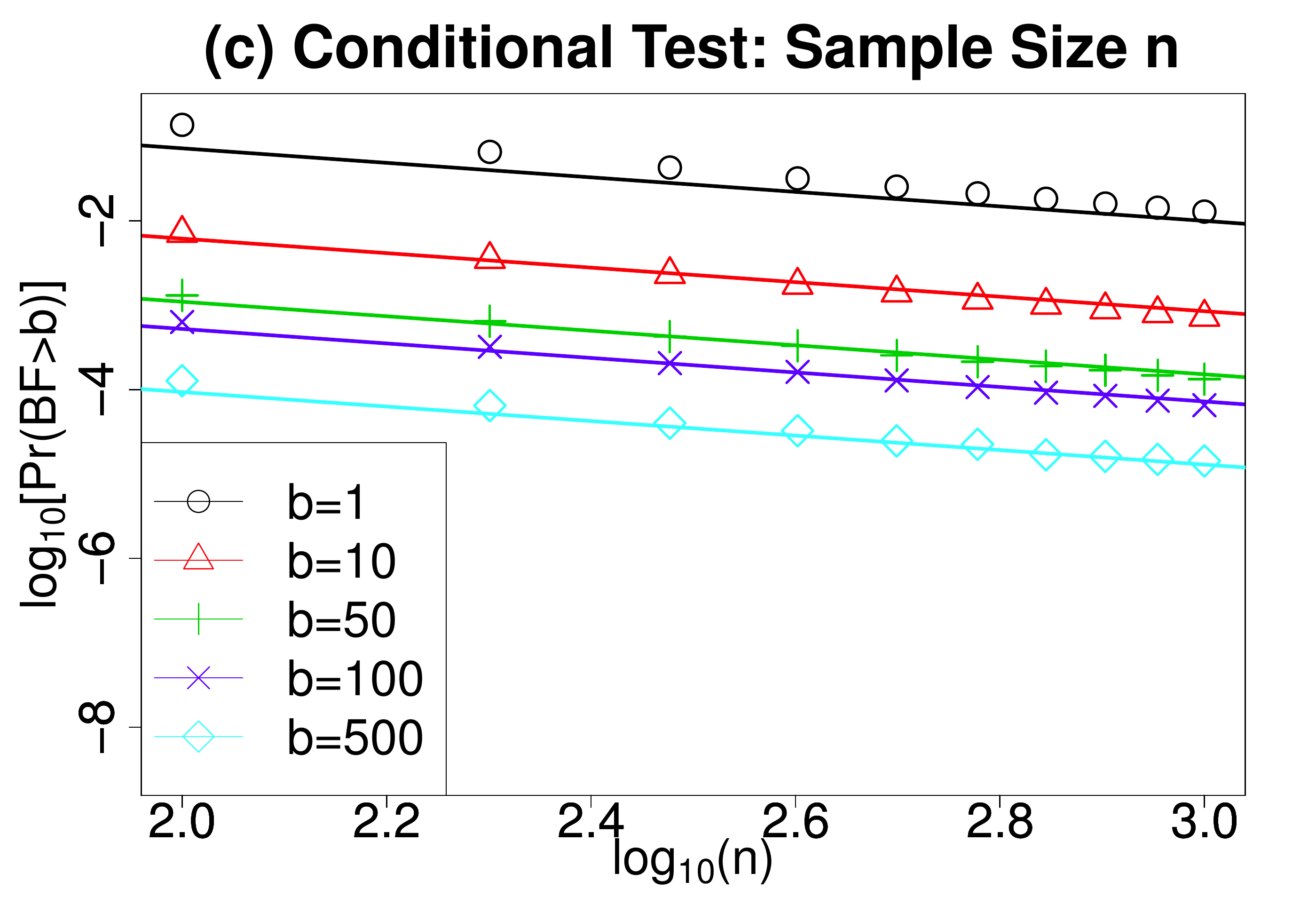}}
\end{subfigure}
\begin{subfigure}{
\includegraphics[width=0.33\textwidth, angle=-90]{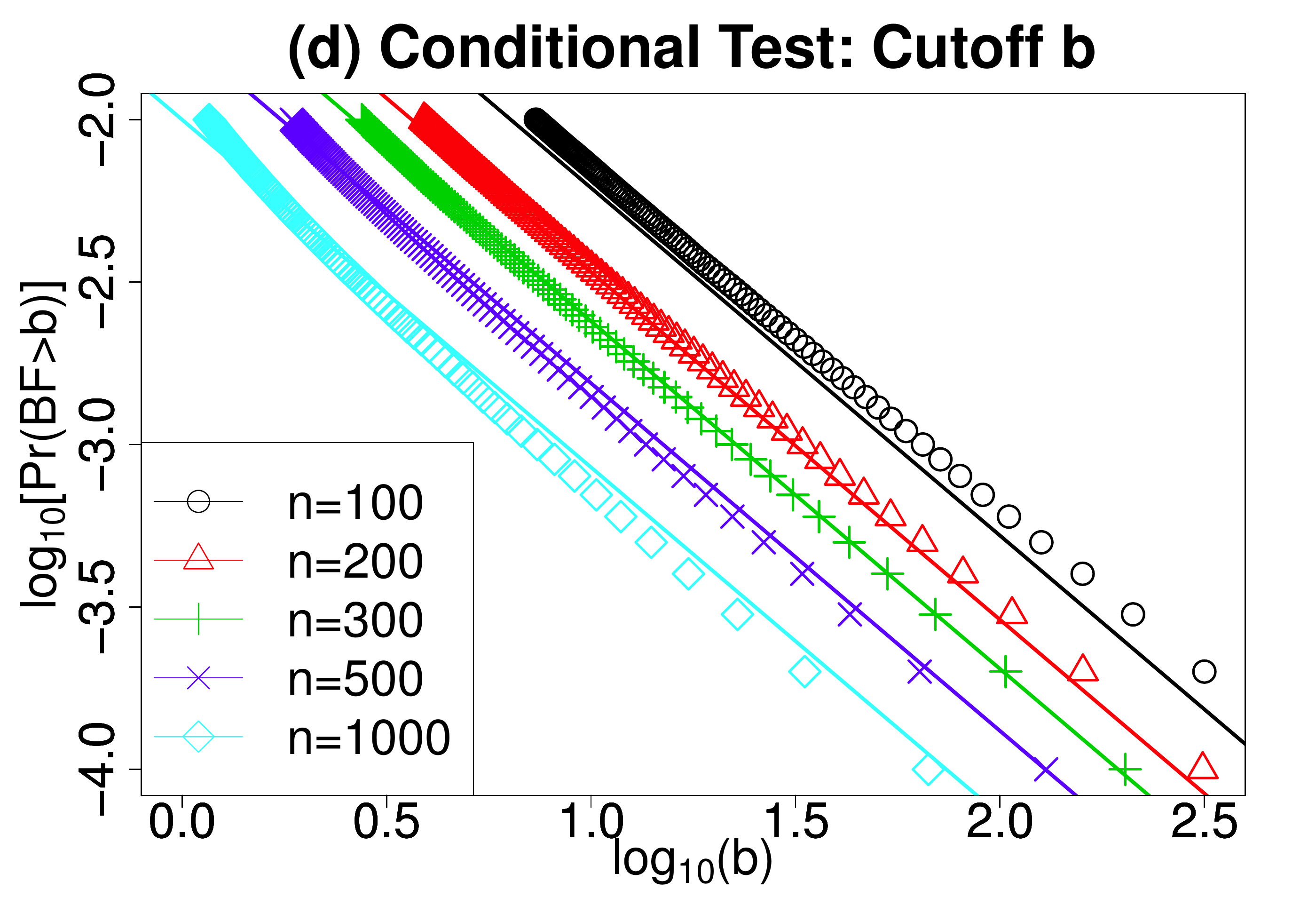}}
\end{subfigure}
\centering
\caption{Empirical fitting of type-I error $\mathrm{Pr}_{\mathrm{shuffle}}\left(\mathrm{BF}\left(X|Y,Z\right) > b\right) $ given sample size $n$ and cutoff $b$ for unconditional and conditional test. Straight lines in (a)-(b) and (c)-(d) are calculated from empirical formula (\ref{eq:fit1}) and (\ref{eq:fit2}), respectively.}
\label{fig:typeI}
\end{figure}

\subsection{Empirical formulas for type-I errors}

Theorem $1$ provides finite-sample  bounds for type-I errors under the conditional
 shuffling scheme, that is,
$$
\mathrm{Pr}_{\mathrm{shuffle}}\left(\mathrm{BF}\left(X|Y,Z\right) > b\right) \leq C_1n^{|Z|(|X|-1)} \min \left\{\frac{1}{(\log(b)+1)n^{\lambda_0-3}}, \frac{1}{b}\right\}.
$$ 
Based on numerical  simulations, we found that the relationship between the value of BF statistic and its significance level can be further refined. Specifically, we simulate observations for both conditional and unconditional tests using the same procedure as described in Section $2.4$ with binary covariate $X$ (or $Z$) and varying sample size $n$. Then, we calculate the BF statistic with hyper-parameters $\lambda_0=1$ and $\alpha_0=1$ on shuffled samples. For unconditional test (\emph{i.e.} $Z \equiv 0$), we obtain the following empirical formula of type-I error given the cutoff of BF statistic $b$ and sample size $n$:
\begin{equation}\label{eq:fit1}
\mathrm{Pr}_{\mathrm{shuffle}}\left(\mathrm{BF}\left(X|Y,Z\equiv0\right) > b\right) \approx
\frac{\gamma_p}{b^{\alpha_p}n^{\beta_p}},
\end{equation}
where $\alpha_p$, $\beta_p$ and $\gamma_p$ only depend on $p$, the proportion of observations with $X=1$. For example, when $p=0.5$, $\alpha_{0.5} \approx 1.12$, $\beta_{0.5} \approx 0.6$ and $\gamma_{0.5} \approx 0.76$ for $|X|=2$ and $|Z|=1$. Figure~\ref{fig:typeI}(a)-(b) illustrate the fitting between empirical formula (\ref{eq:fit1}) and observed values of type-I error when $p=0.5$. Fitted $\alpha_p$, $\beta_p$ and $\gamma_p$ for other values of $p$ are given in the online supplement \citep{jiang2015}.

Moreover, for conditional test, we obtain the following empirical formula of type-I error given the cutoff of BF statistic $b$ and sample size $n$:
\begin{equation}\label{eq:fit2}
\mathrm{Pr}_{\mathrm{shuffle}}\left(\mathrm{BF}\left(X|Y,Z\right) > b\right) \approx
\frac{\gamma_{\bold{f}}}{ b^{\alpha_{\bold{f}}} n^{\beta_{\bold{f}}} },
\end{equation}
where $\alpha_{\bold{f}}$, $\beta_{\bold{f}}$ and $\gamma_{\bold{f}}$ only depend on $\bold{f}$, the vector of observed frequencies for configurations of $(X,Z)$. For example, given $|X|=|Z|=2$ and $\bold{f}=(0.25,0.25,0.25,0.25)$, $\alpha_{\bold{f}} \approx 1.07$, $\beta_{\bold{f}} \approx 0.86$, and $\gamma_{\bold{f}} \approx 3.8$. Figure~\ref{fig:typeI}(c)-(d) illustrate the fitting between empirical formula (\ref{eq:fit2}) and observed values of type-I error. These fitting formulas are useful when one has to deal with many similar hypotheses simultaneously or is interested in very small p-values.

\subsection{Forward stepwise selection based on the BF statistic}

Given a continuous response $Y$ and a set of categorical covariates  $\{X_j\}_{j=1}^m$, variable selection procedures aim to select a subset of covariates indexed by $\mathcal{A}$ such that $\{X_j:j\in\mathcal{A}\}$ are associated with the response $Y$ while the other covariates $\{X_j:j\notin\mathcal{A}\}$ are independent of $Y$ given $\{X_j:j\in\mathcal{A}\}$. Here, we propose to use a forward stepwise procedure based on conditional BF statistic, preceded by an independent screening stage based on unconditional BF statistic. Throughout this paper, we assume that the number of categorical covariates, $m$, is fixed and does not increase with sample size $n$. 

\begin{algorithm}\label{alg:step} \ 
\normalfont 
\begin{itemize}
\item \emph{Independent Screening}: calculate unconditional BF statistic denoted as $\mathrm{BF}\left(X_j\mid Y,Z\equiv0\right)$ for $j=1,\ldots,m$. Let $\mathcal{B}$ denote the index set of covariates with the corresponding BF statistics larger than a pre-specified threshold $b_0$, and $j_0 = \underset{j \in \{1,\ldots,m\}}{\operatorname{argmax}} \left\{\mathrm{BF}\left(X_j|Y,Z\equiv0\right)\right\}$. Proceed if $\mathrm{BF}\left(X_{j_0}|Y,Z\equiv0\right) > b_0$ (\emph{i.e.} $\mathcal{B} \neq \emptyset$).
\item \emph{Forward Stepwise Selection}: let $\mathcal{C}_t$ denote the index set of covariates that have been selected at iteration $t$. Initialize $\mathcal{C}_1={j_0}$ and $Z_1=X_{j_0}$, and iterate the following steps for $t \geq 2$:
\begin{itemize}
\item At iteration $t$ ($t \geq 2$), encode the configurations of selected variables in $\mathcal{C}_{t-1}$ into a ``super'' variable $Z_{t-1}$.
\item Calculate conditional BF statistic $\mathrm{BF}\left(X_j|Y,Z_{t-1}\right)$ for $j \in \mathcal{B}-\mathcal{C}_{t-1}$, and let $j_t = \underset{j \in \mathcal{B}-\mathcal{C}_{t-1}}{\operatorname{argmax}}\left\{\mathrm{BF}\left(X_j|Y,Z_{t-1}\right)\right\}$.
\item Let $\mathcal{C}_t = \mathcal{C}_{t-1} \cup \{j_t\}$ if $\mathrm{BF}\left(X_{j_t}|Y,Z_{t-1}\right) > b_t$. Otherwise, stop and output $\mathcal{C}_{t-1}$.
\end{itemize}
\end{itemize}
\end{algorithm}
\noindent We may decide the threshold $b_t$ at iteration $t$ according to our prior belief in the null hypothesis or a pre-specified interpretation on the relationship between Bayes factor and strength of evidence. For example, \cite{kass1995} viewed a Bayes factor of $>150$ as very strong evidence against the null hypothesis. Alternatively, we can choose the threshold $b_t$ to control for type I errors. Specifically, at each iteration, we estimate the null distribution of the maximum BF statistics under $H_0$ by using a conditional permutation scheme as follows:
\begin{itemize}
\item To generate a permuted data set at iteration $t$, shuffle the observed values of $Y$ within each group of observations indexed by $\{i : z_{t-1,i} = k\}$, independently for $k \in \{1,\ldots,|Z_{t-1}|\}$.
\item Estimate a null distribution of $\mathrm{BF}\left(X_{j_t}|Y,Z_{t-1}\right)$ by calculating the maximum of BF statistics for $j \in \mathcal{B}-\mathcal{C}_{t-1}$ on each permuted data set.
\end{itemize}
Then, we can use the empirical null distribution to calculate a $p$-value for the observed value of $\mathrm{BF}\left(X_{j_t}|Y,Z_{t-1}\right)$, and terminate the iterative variable selection procedure if the $p$-value is larger than a threshold (\emph{e.g.}, $0.05$).

\section{Simulation studies}

\subsection{Unconditional dependence testing}

We first compare different methods in testing unconditional dependence between a binary indicator $X$ and a continuous response $Y$. Note that this testing problem is equivalent to the classic two-sample testing problem.  Methods under  comparison considerations include: the BF statistic with hyper-parameters $\lambda_0=1$ and $\alpha_0=1$ or $2$  (which we call ``BF ($\alpha_0=1$ or $2$)''), dynamic slicing (``DS'') test statistic \citep[see online supplement][for details]{jiang2015}, the Wilcoxon rank-sum test \citep[``rank-sum''; also known as the Mann-Whitney $U$ test;][]{wilcoxon1945,mann1947}, Welch's $t$-test \citep[``$t$-test'';][]{welch1947}, Kolmogorov-Smirnov (``KS'') test and Anderson-Darling (``AD'') test \citep{anderson1952}. The null hypothesis of Welch's $t$-test is that the means of two normally distributed populations are equal (but with possibly unequal variance), and the null hypothesis of rank-sum test is that the probability of an observation from one population exceeding an observation from the second population equals to $0.5$. All other methods test the null hypothesis that the distributions of two populations are the same against a completely general alternative hypothesis that the binary indicator $X$ and the quantity of interest $Y$ are not independent.


\begin{figure}[ht]
\begin{subfigure}{
\includegraphics[width=0.33\textwidth, angle=-90]{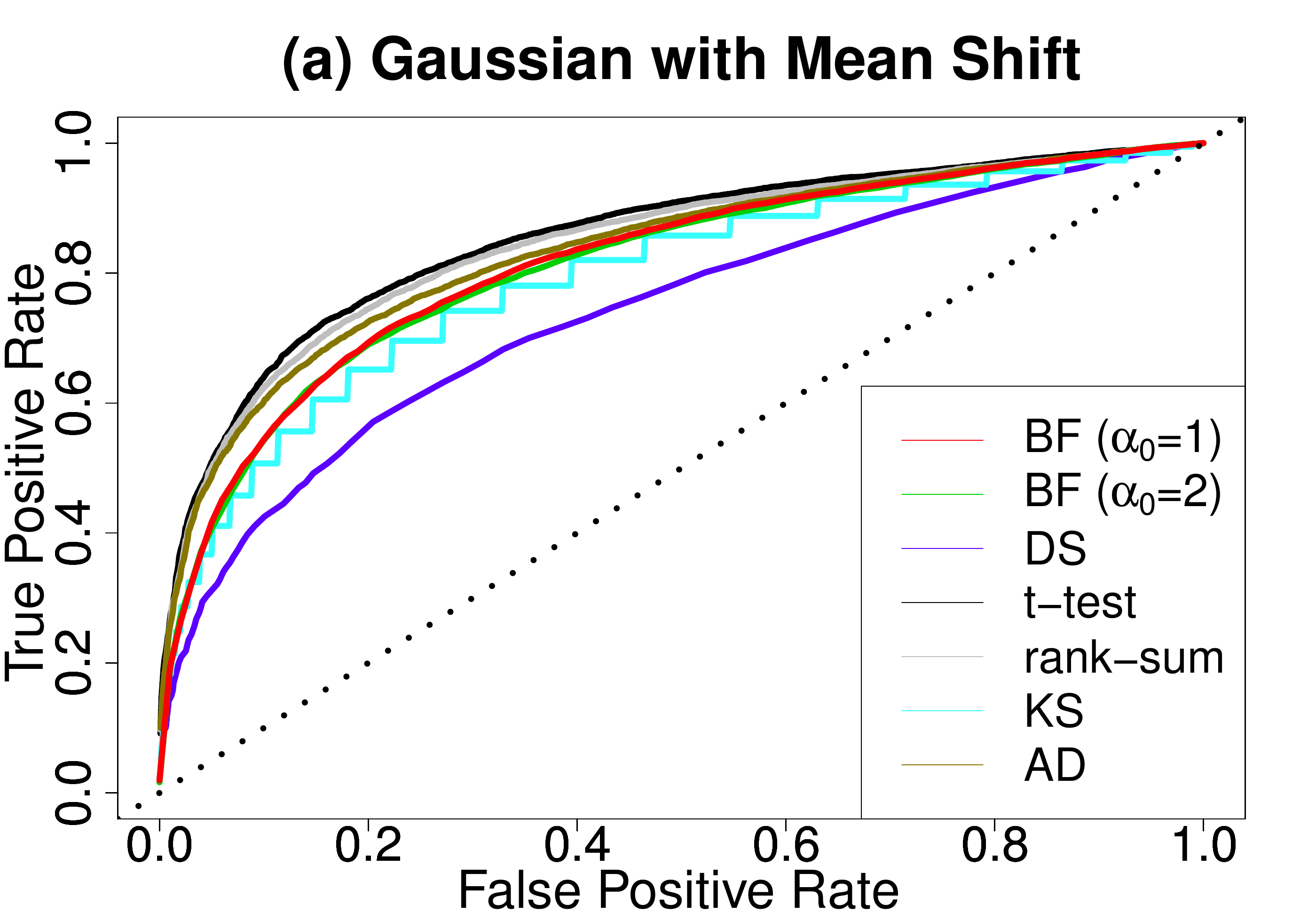}}
\end{subfigure}
\begin{subfigure}{
\includegraphics[width=0.33\textwidth, angle=-90]{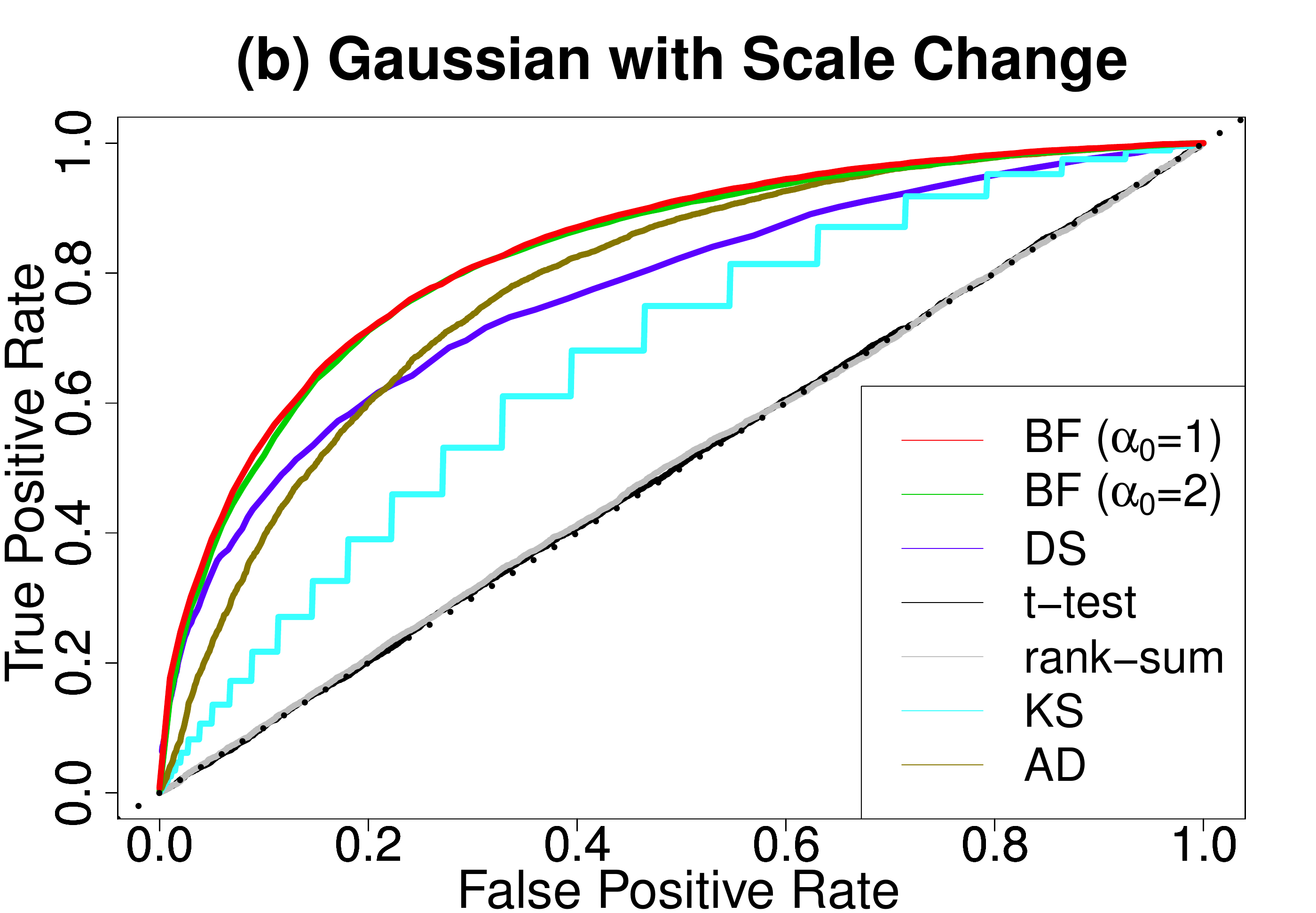}}
\end{subfigure}
\begin{subfigure}{
\includegraphics[width=0.33\textwidth, angle=-90]{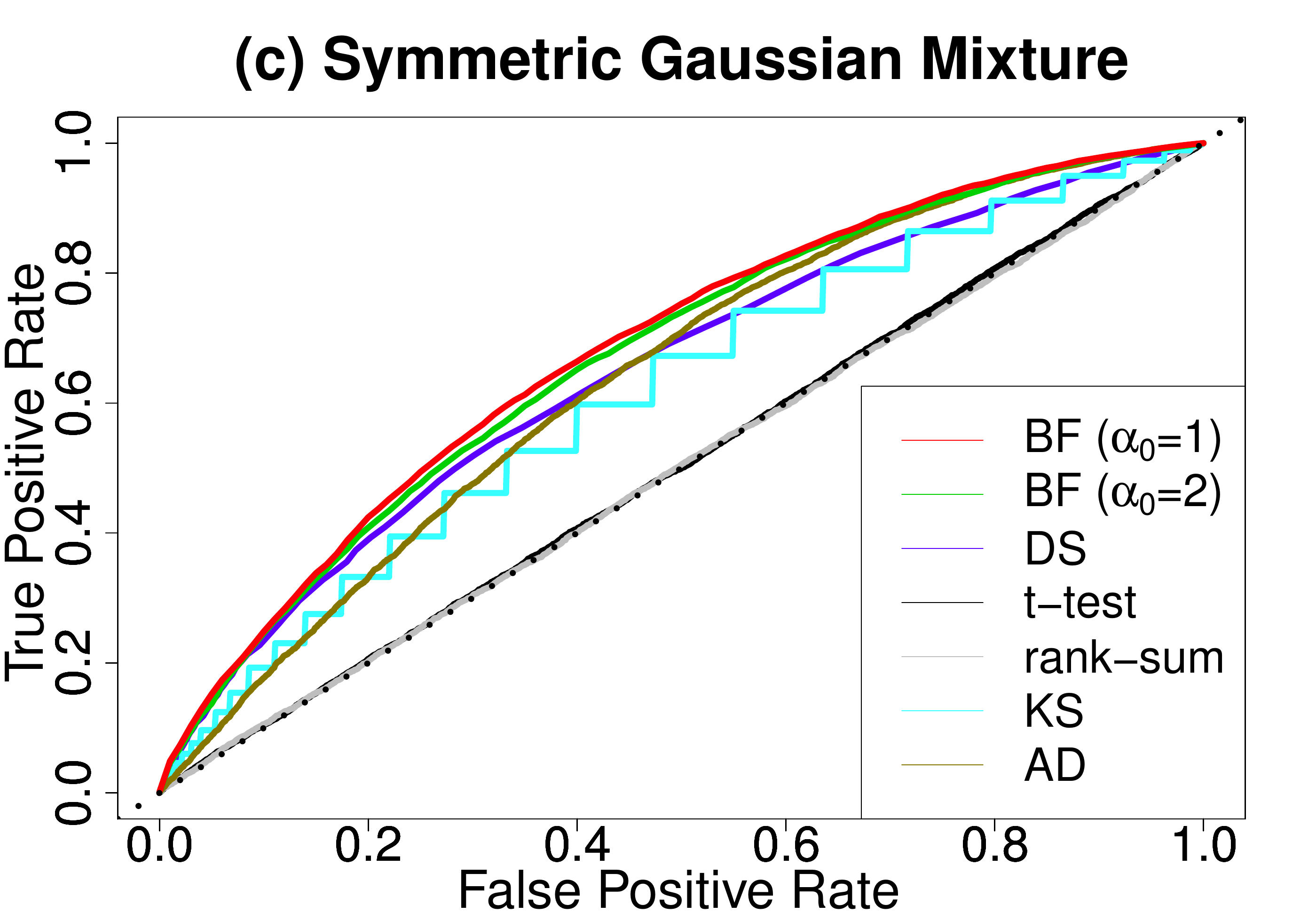}}
\end{subfigure}
\begin{subfigure}{
\includegraphics[width=0.33\textwidth, angle=-90]{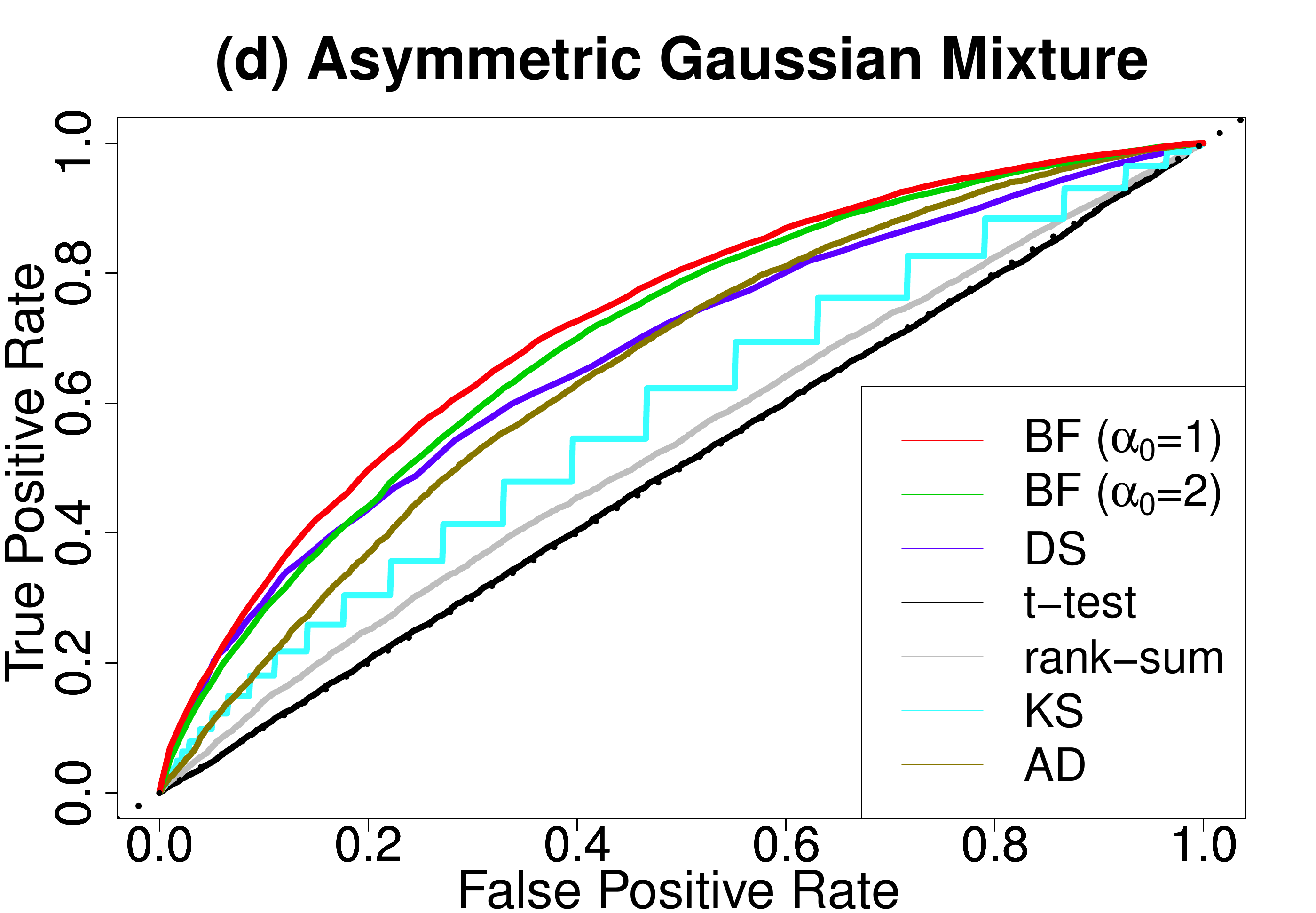}}
\end{subfigure}
\centering
\caption{\label{fig:roc0} The ROC curves that compare true positive rate, the fraction of true positives out of the total actual positives, and false positive rate,  the fraction of false positives out of the total actual negatives, of different methods in Scenarios $1$-$4$ of Section~$3.1$.}
\end{figure}

We generated binary variable $X \sim \text{Bern}(0.5)$, and simulated the continuous variable $Y$ under the alternative hypothesis according to following scenarios with sample size $n=400$:
\begin{itemize}
\item[] Scenario $1$ (Gaussian with mean shift):  
$Y \sim N(-\mu,1) \text{ when } X=0; \text{ and } Y \sim N(\mu,1) \text{ when } X=1$; $\mu = 0.1$.
\item[] Scenario $2$ (Gaussian with scale change):  
$Y \sim N(0,1) \text{ when } X=0; \text{ and } Y \sim N(0,\sigma^2) \text{ when } X=1$; $\sigma =1.2$.
\item[] Scenario $3$ (Symmetric Gaussian mixture):  \\
$Y \sim \text{a mixture of } N(-\mu,1) \text{ and } N(\mu,1)$ $\text{with probabilities $1-\theta$ and }  \\ \text{$\theta$ when } X=0; \text{ and }  Y \sim N\left((2\theta-1)\mu,1+4\theta(1-\theta)\mu^2\right)$ $\text{when } X=1$; $\theta=0.5$ and $\mu = 1.2$.
\item[] Scenario $4$ (Asymmetric Gaussian mixture):  
same as Scenario $3$ except that $\theta=0.9$.
\end{itemize}
The receiver operating characteristic (ROC) curves in Figure~\ref{fig:roc0} illustrate how  true positive rates, the fraction of true positives out of the total actual positives, trade against false positive rates,  the fraction of false positives out of the total actual negatives, of different methods at varying thresholds.
In Scenario $1$, two populations corresponding to $X=0$ and $X=1$ follow Gaussian distributions with different means but the same variance, which satisfies all the parametric assumptions of the two-sample $t$-test. As expected, in Figure~\ref{fig:roc0}(a), Welch's $t$-test achieved the highest power in this scenario, which was followed closely by the rank-sum test and the Anderson-Darling test. The BF statistic had slightly lower power under this scenarios but still outperformed the Kolmogorov-Smirnov test. The dependence test based on dynamic slicing had the lowest power in this case.

When two Gaussian populations have the same mean but different variances (Scenario $2$), the BF statistic had a superior power compared with others in Figure~\ref{fig:roc0}(b). Among the other methods, the Anderson-Darling test and  dynamic slicing (DS) test outperformed the Kolmogorov-Smirnov tests, while the rank-sum test and $t$-test had almost no power under this scenario. 

Scenarios $3$ and $4$ demonstrate the performances of different methods when two populations have both the same mean and the same variance, but different skewness and kurtosis. In both scenarios, the BF statistic with $\alpha_0=1$ achieved the highest power as shown in Figure~\ref{fig:roc0}(c)-(d). The ROC curves of BF statistics with $\alpha_0=1$ and $\alpha_0=2$ were similar in Scenarios $1$-$3$, while the BF statistic with $\alpha_0=1$ had a slightly better performance under Scenario $4$. The BF statistic with $\alpha_0=1$ dominated the dynamic slicing (DS) test statistic in all the four scenarios.

\subsection{Conditional dependence testing}

Next, we compare different methods in testing the conditional dependence between a binary covariate $X$ and a continuous response $Y$ given another binary covariate $Z$. Methods under comparison consideration include: the BF statistic (with $\lambda_0=1$ and $\alpha_0=1$), dynamic slicing (``DS'') statistic, and two-way ANOVA test, which tests for main and interaction effects of $X$ conditioning on $Z$.

\begin{figure}[h]
\centering
\includegraphics[width=0.7\textwidth, angle=-90]{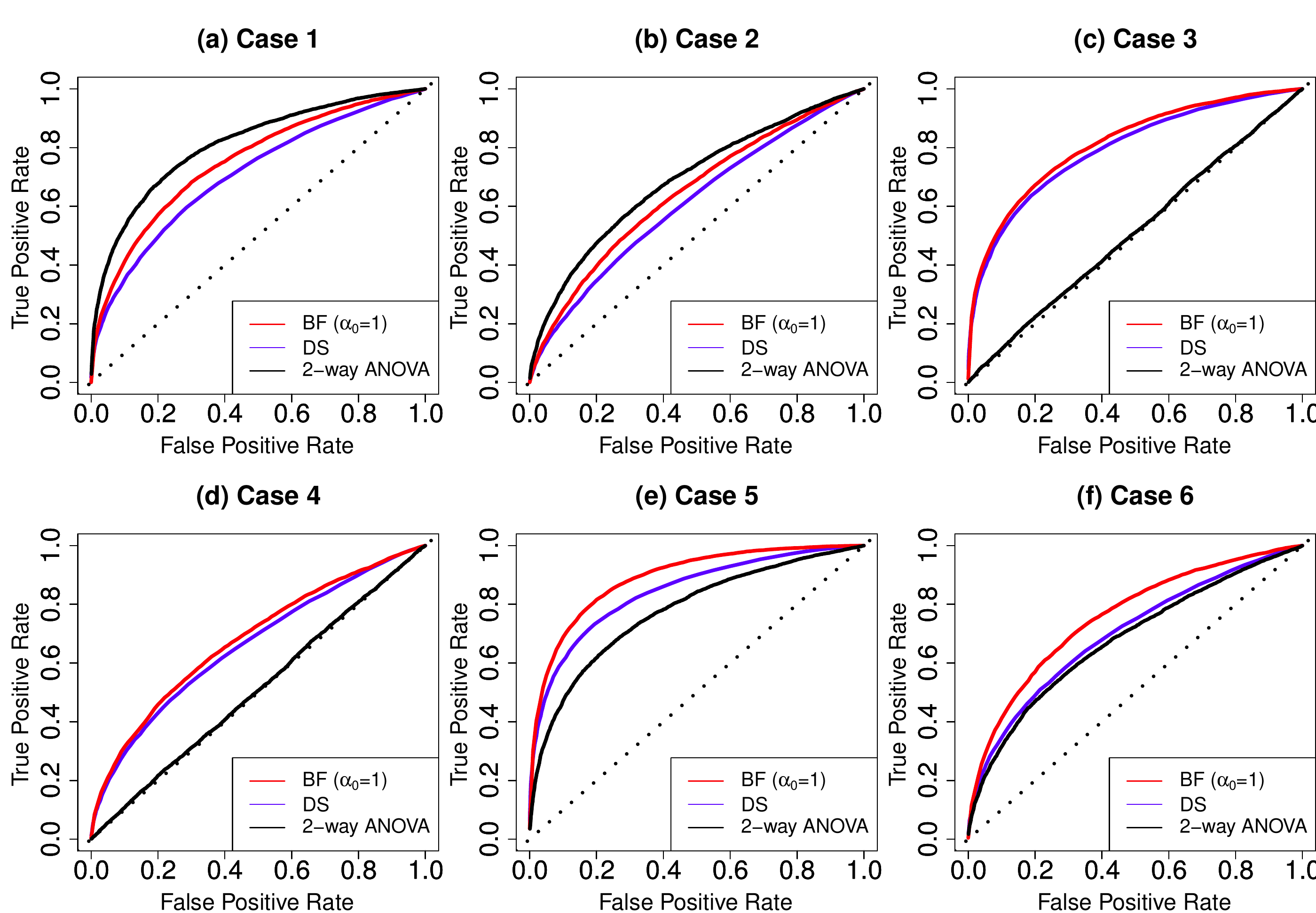}
\caption{The ROC curves that compare true positive rate, the fraction of true positives out of the total actual positives, and false positive rate,  the fraction of false positives out of the total actual negatives, of different methods in Case $1$-$6$ of Section~$3.2$ with uncorrelated covariates $X$ and $Z$.}
\label{fig:roc1}
\end{figure}

\begin{figure}[h]
\centering
\includegraphics[width=0.7\textwidth, angle=-90]{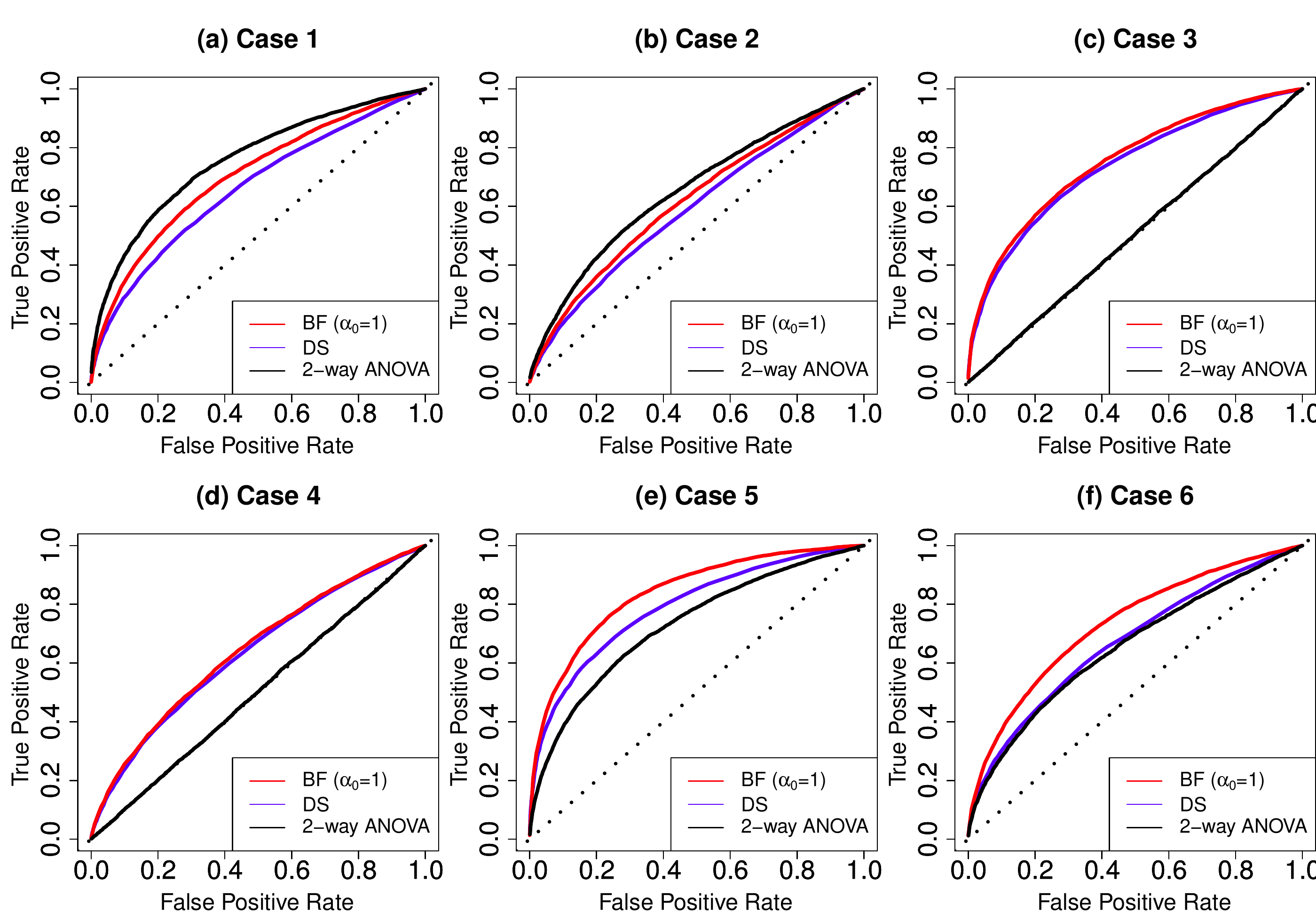}
\caption{The ROC curves that compare true positive rate, the fraction of true positives out of the total actual positives, and false positive rate,  the fraction of false positives out of the total actual negatives, of different methods in Case $1$-$6$ of Section~$3.2$ with correlated covariates $X$ and $Z$.}
\label{fig:roc2}
\end{figure}

In our study, we generate $n=400$ samples with binary covariates $Z \sim \text{Bern}(0.5)$, and conditioning on $Z$, $X|Z=0 \sim \text{Bern}(p_0)$ and $X|Z=1 \sim \text{Bern}(1-p_0)$. We choose $p_0=0.5$ for conditional tests with uncorrelated covariates, and $p_0=0.75$  for conditional tests with correlated covariates. Then, we simulate the response $Y$ according to the following models under the alternative hypothesis:
\begin{itemize}
\item[] Case $1$: $Y = \mu Z + \mu X + \epsilon$; $\epsilon \sim N(0,1)$, $\mu=0.2$.
\item[] Case $2$: $Y = \mu Z X + \epsilon$; $\epsilon \sim N(0,1)$, $\mu=0.2$.
\item[] Case $3$: $Y = \mu Z + \mu X + \epsilon$; $\epsilon \sim \text{Cauchy}(0,1)$, $\mu=0.4$.
\item[] Case $4$: $Y = \mu Z X + \epsilon$; $\epsilon \sim \text{Cauchy}(0,1)$, $\mu=0.4$.
\item[] Case $5$: $Y = \mu Z + \mu X + \epsilon$; $\epsilon \sim N(0,(1+\gamma X)^2)$, $\mu=0.2, \gamma=0.2$.
\item[] Case $6$: $Y = \mu Z X + \epsilon$; $\epsilon \sim N(0,(1+\gamma Z X)^2)$, $\mu=0.2, \gamma=0.2$.
\end{itemize}
Our goal here is to test whether $X$ is independent of $Y$ given $Z$. The ROC curves of different methods under Case $1$-$6$ with uncorrelated ($p_0=0.5$) and correlated ($p_0=0.75$) $X$ and $Z$ are given in Figure~\ref{fig:roc1} and Figure~\ref{fig:roc2}, respectively.

In Cases $1$ and $2$, the two samples were generated from homoscedastic normal distribution with either linear combination or multiplicative interaction of covariates, which  satisfies all the parametric assumptions of the two-way ANOVA test. As we have expected, the two-way ANOVA test achieved highest power  in Figure~\ref{fig:roc1}(a)-(b) and Figure~\ref{fig:roc2}(a)-(b), followed by the BF and DS test statistics. 

However, as shown in Figure~\ref{fig:roc1}(c)-(d) and Figure~\ref{fig:roc2}(c)-(d), when the two samples were generated from Cauchy distribution, the two-way ANOVA test was completely powerless in Case $3$ and $4$, while the BF and DS statistics had considerable powers with Cauchy noises.

Cases $5$ and $6$ illustrate the scenarios when the response has heteroscedastic variances depending on covariates. In both cases, the BF statistic achieved better powers than the DS and two-way ANOVA test. Among all the conditional dependence testing scenarios we have considered, the BF test statistic always outperformed the DS test statistic. The relative performances of different methods were consistent whether covariates $X$ and $Z$ are correlated or not.

\subsection{Interaction detection on synthetic QTL data}

Traditional QTL studies are based on linear regression models \citep{lander1989} in which each (continuous) trait variable is regressed against each (discrete) marker variable. The $p$-value of the regression slope is reported as a measure of significance for association. \cite{storey2005} developed a stepwise regression method to search for pairs of markers that are associated with the gene expression quantitative trait. This procedure, however, tends to miss QTL pairs with small marginal effects but a strong interaction effect. 

In this section, we compare the proposed variable selection method based on the BF statistic with the stepwise regression (SR) method in identifying genetic markers with interaction effects in synthetic QTL data sets. Using the R package \emph{qtl}, we generated $100$ binary markers with sample size $n=400$ such that adjacent markers are correlated with each other, and then, we randomly select two markers and simulate quantitative traits according to Cases $1$-$6$ in the previous section. We evaluate the performance of the BF and SR methods using the following procedure. First, in the screening step, we calculate the unconditional test statistic for each marker and obtain a list of candidate markers with test statistic above a given threshold $T_1$. Second, conditioning on the most significant candidate marker, we select other candidate markers with conditional test statistics above another threshold $T_2$. Finally, we vary the thresholds $T_1$ and $T_2$ simultaneously to generate the ROC curves in Figure~\ref{fig:roc3}.

\begin{figure}[h]
\centering
\includegraphics[width=0.7\textwidth, angle=-90]{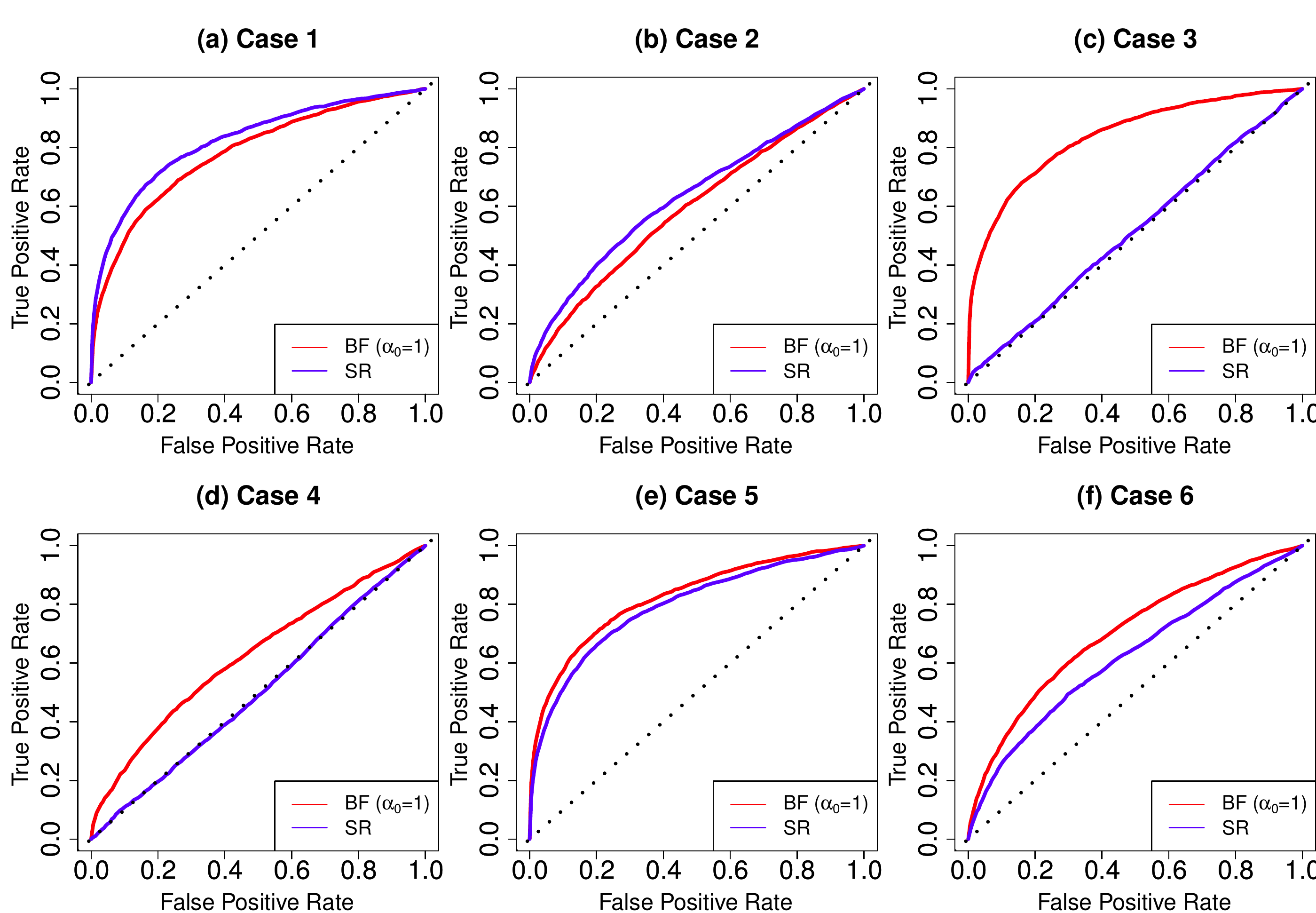}
\caption{The ROC curves that compare true positive rate, the fraction of the true gene-marker pairs detected, and false positive rate, the fraction of unrelated gene-marker pairs falsely selected, of stepwise regression (SR) and BF statistic on synthetic QTL data sets in Section~$3.3$.}
\label{fig:roc3}
\end{figure}

As we can see from Figure~\ref{fig:roc3}, the SR  had a better power when its underlying assumptions, \emph{i.e.}, linearity, normality, and homoscedasticity, were satisfied as in Cases $1$ and $2$. However, with the presence of extreme values or heteroscedastic effects in Cases $3$-$6$, the BF statistic was much more powerful than the SR.

\section{Application to QTL study in mouse}

\cite{burke2012} measured a mouse population for complex adult phenotypes, related to body size and bone structure, and conducted a genome-wide search for QTLs that are marginally associated with quantitative phenotypes. Each mouse in the population was genotyped at $558$ biallelic loci, $i.e.$, binary genetic markers. For each trait, \cite{burke2012} performed a single-locus genome-wide search using one-way ANOVA model and permutation-based test of significance. 

We applied the screening and the BF-based  stepwise selection procedure proposed in Section $2.6$ to search for effective loci associated with two quantitative traits, femur length and vertebra length. At the screening step and each forward selection step, we permuted sample labels conditioning on the observed values of previously selected QTLs, and in each permuted data set, we recorded the maximum value of BF statistics among all the candidate QTLs. Then, a \emph{genome-wide} $p$-value is evaluated by comparing the observed BF statistic with these maximum BF values from $1000$ permuted data sets. In the screening step, we retained $35$ and $49$ loci, respectively for femur length and vertebra length, with unconditional BF value larger than $10$ (corresponding to genome-wide $p$-value of $0.03$ and $0.05$) as candidate QTLs for forward stepwise selection. QTLs identified through stepwise selection on each trait, together with their BF values and genome-wide \textit{p}-values in forward selection steps, as well as relationships with significant loci found in 
\cite{burke2012}, are given in Table~\ref{tab:femur} and~\ref{tab:verte}.

Table~\ref{tab:femur} shows femur length QTLs that are selected from the first $5$ iterations of the proposed stepwise method, together with their BF values and genome-wide $p$-values at each forward selection step. 
\cite{burke2012} reported the same $5$ genomic regions as marginally associated with the trait and their genome-wide $p$-values from the paper are given in Table~\ref{tab:femur}. Using a cutoff of $0.05$ for $p$-values, the BF-based stepwise procedure was terminated after the third iteration, \emph{i.e.}, we could not reject the null hypothesis that \emph{D5Mit25} and mouse femur length are conditionally independent given the previously selected loci rs3091203, \emph{D2Mit285} and rs3657845. This is confirmed on an independent replicate population with femur length measurement and genotypes on a subset of loci ($356$ of $558$ loci) provided by \cite{burke2012}. Although \emph{D5Mit25} (located at CH5$\cdot$114) was not genotyped in the replicate population, genotypes of its neighboring locus rs13478469 (also located at CH5$\cdot$114) were available. Note that in the original population, both \emph{D5Mit25} and rs13478469 have genome-wide $p$-values smaller than $0.001$ according to \cite{burke2012}, but rs13478469 was not reported in \cite{burke2012} due to its adjacency and high correlation (about $0.95$) with \emph{D5Mit25}. In the replicate population, $3$-way ANOVA test shows that the top $3$ QTLs in Table~\ref{tab:femur} all have significant main effects with $p$-values $<0.001$. On the other hand, rs13478469 does not have significant main effect or interaction effects with other $3$ QTLs ($p$-values $>0.1$) according to $4$-way ANOVA test on the replicate population. These results from an independent replicate population are consistent with the conclusions of our testing procedure applied to the original population.


\begin{table}[h]
\caption{\label{tab:femur} Identified QTLs (ranked according to their orders in forward stepwise selection) associated with mouse femur length, their Bayes factor values and corresponding genome-wide $p$-values, and relationships with significant loci reported in \cite{burke2012}. Genomic location is in \textit{Chromosome$\cdot$Mb} format.}
\vspace{5mm}
\centering
\fbox{
\begin{tabular}{ c c c c c c}
   QTL & Location & Bayes Factor & $p$-value & Reported in \cite{burke2012} \\ 
    \hline
    rs3091203 & CH13$\cdot$22 & $2.7\times10^7$ & $<0.001$ & $p$-value $<0.001$ \\
    \emph{D2Mit285} & CH2$\cdot$152  & $1.8\times10^5$ & $<0.001$ & CH2$\cdot$157 (rs4223627), $p$-value $<0.001$ \\
    rs3657845 & CH17$\cdot$17 & $63.7$ & $0.004$ & $p$-value $=0.011$ \\
    \emph{D5Mit25} & CH5$\cdot$114 & $2.3$ & $0.192$ & $p$-value $<0.001$ \\
    \emph{D9Mit110} & CH9$\cdot$91 & $46.9$ & $0.062$ & $p$-value $<0.001$ \\
  \end{tabular}
}
\end{table}

\begin{table}[h]
\caption{\label{tab:verte} Identified QTLs (ranked according to their orders in forward stepwise selection) associated with mouse vertebra length, their Bayes factor values and corresponding genome-wide $p$-values, and relationships with significant loci reported in \cite{burke2012}. Genomic location is in \emph{Chromosome$\cdot$Mb} format.}
\centering
\fbox{
\begin{tabular}{ c c c c c c}
   QTL & Location & Bayes Factor & $p$-value & Reported in \cite{burke2012} \\ 
    \hline
    rs4222738 & CH1$\cdot$158 & $7.0\times10^9$ & $<0.001$ & $p$-value $<0.001$ \\
    \emph{D1Mit105} & CH1$\cdot$162 & $4.6\times10^7$ & $<0.001$ & CH1$\cdot$166 (rs4222769), $p$-value $<0.001$ \\
    \emph{D2Mit58} & CH2$\cdot$108 & $1.6\times10^3$ & $0.001$ & CH2$\cdot$111 (rs3023543), $p$-value $=0.006$ \\
    \emph{D16Mit36} & CH16$\cdot$31 & $991.8$ & $0.001$ & {\bf Not reported} \\
    \emph{D7Mit76} & CH7$\cdot$18 & $247.9$ & $0.019$ & $p$-value $0.0026$ \\
    rs13481706 & CH13$\cdot$16 & $2.8\times10^5$ & $0.039$ & $p$-value $0.019$ \\
  \end{tabular}
}
\end{table}

In Table~\ref{tab:verte}, the proposed forward stepwise procedure based on the BF statistic detected $6$ QTLs associated with vertebra length under a significance level of $0.05$, which include all of the $5$ genomic regions (either the locus itself or the neighboring locus located next to it) reported in \cite{burke2012}. Besides, our analysis identified an additional locus, \emph{D16Mit36}. From the first plot in Figure~\ref{fig:box}, we can see that although the distributions of vertebra length given two alleles of \emph{D16Mit36} have similar means (one-way ANOVA test of equal means has $p$-value $=0.19$), the variances are quite different (an F-test of equal variances has $p$-value $=7.16\times10^{-5}$). Because of its heteroscedastic effect, the screening based on the BF statistic was able to retain \emph{D16Mit36} as a candidate QTL, while one-way ANOVA test missed the locus completely. Further analysis shows that QTLs \emph{D16Mit36} and \emph{D2Mit58}, which was identified in the previous step, have positive \emph{epistasis} effect. From Figure~\ref{fig:box}, we can see that \emph{D16Mit36} and \emph{D2Mit58}  have non-additive effects, that is, the effect of \emph{D2Mit58} is larger when \emph{D16Mit36} has allele B6. Two-way ANOVA test of interaction effect between \emph{D16Mit36} and \emph{D2Mit58} has a $p$-value of $=0.001$. This example demonstrates that the proposed selection procedure based on the BF statistic is particularly effective in detecting QTLs with interaction and heteroscedastic effects.

\begin{figure}[ht]
\centering
\includegraphics[width=0.6\textwidth, angle=-90]{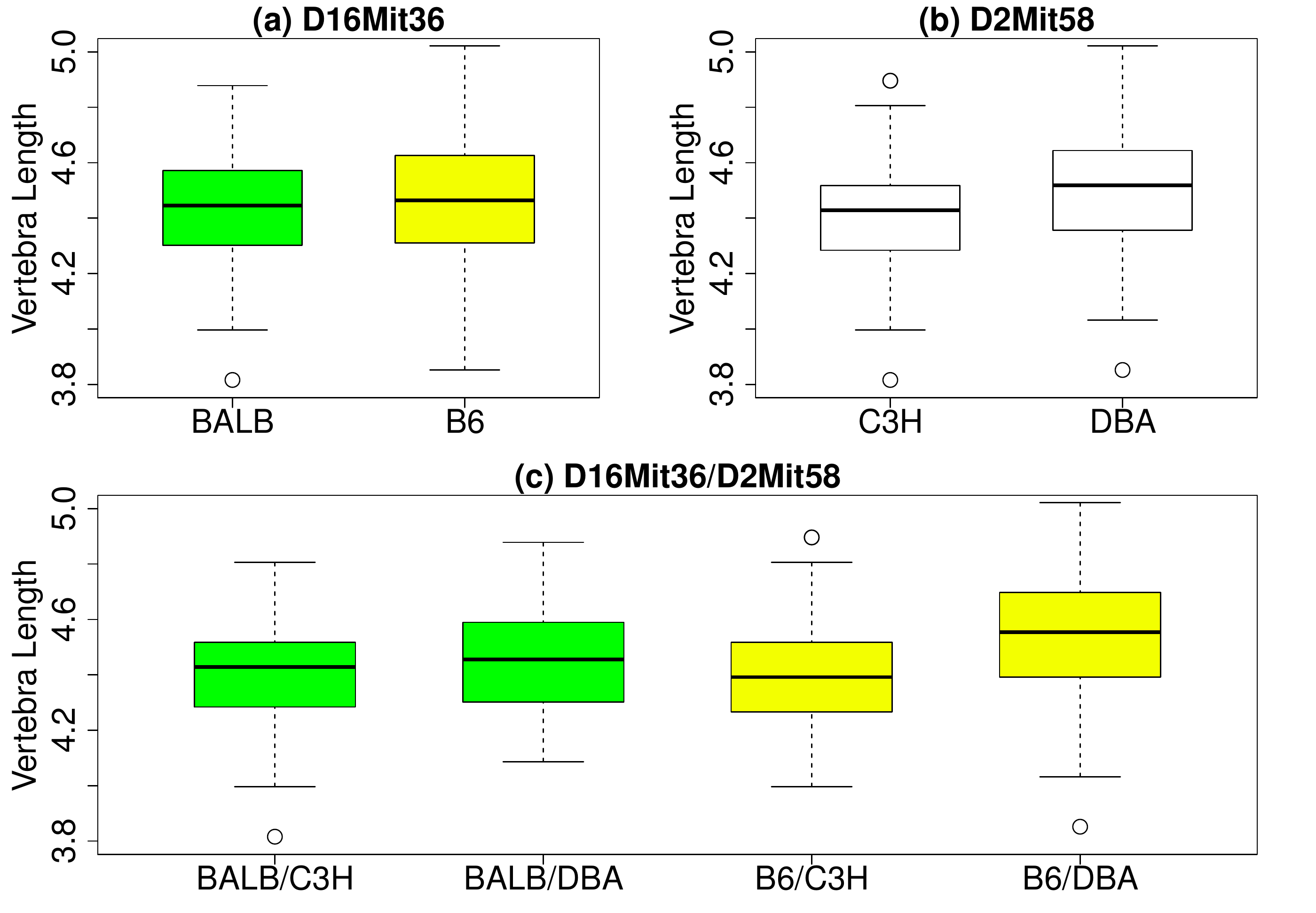}
\caption{Box-plots showing the heteroscedastic effect on vertebra length given two alleles (BALB and B6) of \emph{D16Mit36} (a), the mean shift effect on vertebra length given two alleles (C3H and DBA) of \emph{D2Mit58} (b), and the epistasis effect between \emph{D16Mit36} and \emph{D2Mit58} (c).}
\label{fig:box}
\end{figure}

\section{Discussion}

We have developed a non-parametric dependence testing method for categorical covariates and continuous response, and implemented the proposed method in R package \emph{bfslice}, which can be downloaded from \url{http://www.people.fas.harvard.edu/~junliu/BF/} or requested from the authors directly. As a dependence testing tool, the proposed Bayes factor-based statistic achieves a higher statistical power compared with traditional non-parametric methods such as the Kolmogorov-Smirnov test, and is more robust to outliers and various distributional assumptions compared with classical ANOVA based approaches. Furthermore, the stepwise variable selection method based on the BF statistic is particularly effective in detecting covariates with interaction or heteroscedastic effects, especially when the combined number of covariate categories is relatively large compared with sample size. Theoretically, we proved upper bounds on p-values (type-I errors) of the BF statistic under a variety of null hypothesis assumptions, and showed that the proposed BF statistic asymptotically grows to infinity at an exponential rate under the alternative hypothesis and with proper choices of hyper-parameters.  We also  fitted a fairly accurate empirical formula for the type-I error of any given BF cutoff value. But a theoretical derivation of its exact form remains an open question. 

The method described in this paper can be easily used to deal with categorical or discrete ordinal response variables. For categorical response, different response categories naturally define the ``slicing'' scheme, and for discrete ordinal response (or continuous response with ties), we can arbitrarily rank observations with the same value of response and only allow slicing between ranked observations that have different observed response values. A potential research direction is to extend the Bayes factor approach for variable selection with continuous covariates under the sliced inverse regression framework \citep{jiang2014a} 

For applications with multivariate response, we can further generalize the concept of ``slices'' to unobserved clusters (aka ``partitions'' ) of samples, and model the distribution of response and covariates independently given hidden cluster labels. Combined with a Markov Chain Monte Carlo strategy, we are currently developing a Bayesian partition procedure for detecting expression quantitative trait loci (eQTLs) with variable selection on both responses (gene expression) and covariates (genetic variations). 

\appendix

\section{Proofs}\label{app:proofs}

\subsection{Proof of Theorem~\ref{thm:bf1}}\label{app:thm1}

\begin{proof}[Theorem~\ref{thm:bf1}]
Without loss of generality, we assume that the observations have been ranked according to their values of $Y$, that is $y_i = y_{(i)}$. For any $t \in \{1,\ldots,n\}$, let $n_{j,k}^{(1:t)} = |\{i : X_i=k, Z_i=j, i=1,\ldots,t\}|$ and $n_{j}^{(1:t)}=\sum_{k=1}^{|X|}n_{j,k}^{(1:t)}$. We use $\text{Pr}_{\text{shuffle}}\left( \left. \{x_i\}_{i=1}^t \right| \{z_i\}_{i=1}^t \right)$ to denote the probability of observing $\{X_i=x_i\}_{i=1}^t$ given $\{z_i\}_{i=1}^t$ under a conditional permutation scheme, which samples the value of $X_i$ ($i=1,\ldots,t$) without replacement from $\{x_s: z_s = z_i, s = 1,\ldots,n\}$. Conditioning on $\{n_{j,k}:j=1,\ldots,|Z|,k=1,\ldots,|X|\}$, we have
\begin{eqnarray}\nonumber
& & \text{Pr}_{\text{shuffle}}\left( \left. \{x_i\}_{i=1}^t \right| \{z_i\}_{i=1}^t \right) \\\nonumber
& = & \prod_{j=1}^{|Z|} \frac{\prod_{k=1}^{|X|}\mathbb{I}\left(n_{j,k}^{(1:t)} \leq n_{j,k}\right)}{\dbinom{n_j^{(1:t)}}{{n_{j,1}^{(1:t)},\ldots,n_{j,|X|}^{(1:t)}}}} \frac{\dbinom{n_j^{(1:t)}}{n_{j,1}^{(1:t)},\ldots,n_{j,|X|}^{(1:t)}} \dbinom{n_j-n_j^{(1:t)}}{n_{j,1}-n_{j,1}^{(1:t)},\ldots,n_{j,|X|}-n_{j,|X|}^{(1:t)}}}{\dbinom{n_j}{n_{j,1},\ldots,n_{j,|X|}}}  \\\nonumber
& \leq & \prod_{j=1}^{|Z|} \frac{1}{\dbinom{n_j^{(1:t)}}{n_{j,1}^{(1:t)},\ldots,n_{j,|X|}^{(1:t)}}} = \prod_{j=1}^{|Z|} \left[\frac{\prod_{k=1}^{|X|}\Gamma\left(n_{j,k}^{(1:t)}+1\right)}{\Gamma\left(n_j^{(1:t)}+1\right)}\right],
\end{eqnarray}
where $\mathbb{I}\left(A\right) = 1$ if event $A$ is true and $0$ otherwise, and the inequality follows from the fact that both the indicator functions and the generalized hypergeometric probability are no greater than one.

Given $\alpha_0 \leq |X|$, the Dirichlet prior density of $p_j$, 
$$
f(p_j) = \Gamma\left(\alpha_0\right) \prod_{k=1}^{|X|} \left[p_{j,k}^{\frac{\alpha_0}{|X|}-1}\left/\Gamma\left(\frac{\alpha_0}{|X|}\right.\right)\right] \geq m_0
$$ 
for some $m_0>0$ and $j = 1, \ldots, |Z|$. Let $\psi_{s,t} \equiv \Psi\left(x_s,\ldots,x_t\right)$ ($1 \leq s \leq t \leq n$) denote the probability of observing $\{X_i=x_i\}_{i=s}^t$ given $\{z_i\}_{i=s}^t$ under the Dirichlet priors on $p_j$'s in (\ref{eq:dir}). Then,
\begin{eqnarray}\nonumber
\Psi\left(x_1,\ldots,x_t\right) &=& \int\ldots\int \left[\prod_{j=1}^{|Z|} \Gamma\left(\alpha_0\right) \prod_{k=1}^{|X|}  \frac{p_{j,k}^{n_{j,k}^{(1:t)}+\frac{\alpha_0}{|X|}-1}}{\Gamma\left(\frac{\alpha_0}{|X|}\right)} \right] \prod_{j=1}^{|Z|} dp_j \\\nonumber
&\geq& m_0^{|Z|}\prod_{j=1}^{|Z|} \left[\frac{\prod_{k=1}^{|X|}\Gamma\left(n_{j,k}^{(1:t)}+1\right)}{\Gamma\left(n_j^{(1:t)}+|X|\right)}\right] \\\nonumber
 & \geq & (1/C_0) n^{-|Z|(|X|-1)}\prod_{j=1}^{|Z|} \left[\frac{\prod_{k=1}^{|X|}\Gamma\left(n_{j,k}^{(1:t)}+1\right)}{\Gamma\left(n_j^{(1:t)}+1\right)}\right],
\end{eqnarray}
where we assume $n \geq |X|$ and $C_0 = (2^{|X|}/m_0)^{|Z|} > 0$. Thus,
\begin{equation}\label{eq:cover}
\text{Pr}_{\text{shuffle}}\left( \left. \{x_i\}_{i=1}^t \right| \{z_i\}_{i=1}^t \right) \leq C_0 n^{|Z|(|X|-1)} \Psi\left(x_1,\ldots,x_t\right).
\end{equation}

Given a slicing scheme $S(\cdot)$ with $|S| \geq 2$, we use $n_h$ to denote the number of observations with $S(y_i) = h$. Assume that both observations and slice indices have been sorted, that is, $y_i = y_{(i)}$ and $S(y_i) \leq S(y_j)$ for $i < j$. Then,
$$
\text{BF}\left(X|S(Y),Z\right) =  \frac{\text{Pr}_{H_1}\left(X|S(Y),Z\right)}{\text{Pr}_{H_0}\left(X|Z,Y\right)} = \prod_{h=2}^{|S|} \Delta\left(\sum_{l=1}^{h-1} n_{l}+1, n_h\right),
$$
where
\begin{equation}\label{eq:delta}
\Delta(s,t) = \frac{\psi_{1,s-1}\psi_{s,t}}{\psi_{1,t}} = \frac{\Psi\left(x_1,\ldots,x_{s-1}\right)\Psi\left(x_s,\ldots,x_t\right)}{\Psi\left(x_1,\ldots,x_t\right)},
\end{equation}
 for $2 \leq s \leq t \leq n$. So
\begin{eqnarray}\nonumber
& &\text{BF}\left(X|S(Y),Z\right) \pi_0^{|S|-1}(1-\pi_0)^{n-|S|} \\\nonumber
& = & \frac{(1-\pi_0)^{n-1}}{\left(1-1/n\right)^{n-1}} (1/n)^{|S|-1}(1-1/n)^{n-|S|} \text{BF}\left(X|S(Y),Z\right) \left[(n-1)\frac{\pi_0}{1-\pi_0}\right]^{|S|-1} \\\nonumber
& \leq & e (1/n)^{|S|-1}(1-1/n)^{n-|S|} \prod_{h=2}^{|S|} \left[ (n-1)\frac{\pi_0}{1-\pi_0}
\Delta\left(\sum_{l=1}^{h-1} n_{l}+1, n_h\right)\right] \\\nonumber
& \leq &  (1/n)^{|S|-1}(1-1/n)^{n-|S|} \prod_{h=2}^{|S|} \left[ ne \frac{\pi_0}{1-\pi_0}
\Delta\left(\sum_{l=1}^{h-1} n_{l}+1, n_h\right) \right]
\end{eqnarray}

According to (\ref{eq:bf}),
$$
\text{BF}\left(X|Y,Z\right) = \sum_{S(Y) \in \Omega_n(S)} \text{BF}\left(X|S(Y),Z\right) \pi_0^{|S|-1}(1-\pi_0)^{n-|S|}.
$$
For any $b \geq 1$, if $\max_{2 \leq s \leq t \leq n} \left[ ne\frac{\pi_0}{1-\pi_0} \Delta\left(s,t\right) \right] \leq \log(b)+1$, then 
$$
\text{BF}\left(X|Y,Z\right) \leq \sum_{S(Y) \in \Omega_n(S)}\left(\frac{\log(b)+1}{n}\right)^{|S|-1}\left(1-\frac{1}{n}\right)^{n-|S|} = \left(1+\frac{\log(b)}{n}\right)^{n-1} \leq b,
$$ 
Thus,
\begin{equation}\label{eq:bound}
\text{Pr}_{\text{shuffle}}\left(\text{BF}\left(X|Y,Z\right) > b\right) \leq \text{Pr}_{\text{shuffle}}\left( \max_{2 \leq s \leq t \leq n} \left[ ne\frac{\pi_0}{1-\pi_0} \Delta\left(s,t\right) \right] > \log(b)+1\right).
\end{equation}

When $H_0$ is true, for any $\delta>0$, we have
\begin{eqnarray}\nonumber
& & \text{Pr}_{\text{shuffle}}\left( \Delta\left(s,t\right) > \delta \right) \\\nonumber
& \leq & \sum_{\left(x_1,\ldots,x_t\right)} \delta^{-1} \Delta(s,t) \text{Pr}_{\text{shuffle}}\left( \left. \{x_i\}_{i=1}^t \right| 
\{z_i\}_{i=1}^t \right) \\\nonumber
& \leq & \sum_{\left(x_1,\ldots,x_t\right)} \delta^{-1} \Delta(s,t) C_0 n^{|Z|(|X|-1)} \Psi\left(x_1,\ldots,x_t\right) \\\nonumber
& = & \delta^{-1} C_0 n^{|Z|(|X|-1)} \sum_{\left(x_1,\ldots,x_{s-1}\right)} \Psi\left(x_1,\ldots,x_{s-1}\right) \sum_{\left(x_s,\ldots,x_t\right)} \Psi\left(x_s,\ldots,x_t\right) \\\nonumber
& = & \delta^{-1} C_0n^{|Z|(|X|-1)},
\end{eqnarray}
where $2 \leq s \leq t \leq n$ and $\sum_{\left(x_s,\ldots,x_t\right)}$ denotes the sum over all possible $|X|^{t-s+1}$ configurations of $\left(X_s,\ldots,X_t\right)$. The second inequality follows from (\ref{eq:cover}), the second to the last equality follows from (\ref{eq:delta}) and the last equality follows from the fact that 
$$
\sum_{\left(x_1,\ldots,x_{s-1}\right)} \Psi\left(x_1,\ldots,x_{s-1}\right) = \sum_{\left(x_s,\ldots,x_t\right)} \Psi\left(x_s,\ldots,x_t\right)=1
$$ for $1 \leq s \leq t \leq n$. By letting $\delta = \frac{\log(b)+1}{n\pi_0e/(1-\pi_0)}$ and  $\log\left(\pi_0/(1-\pi_0)\right) = -\lambda_0 \log(n)$, we have
\begin{eqnarray}\nonumber
\text{Pr}_{\text{shuffle}}\left(ne\frac{\pi_0}{1-\pi_0} \Delta\left(s,t\right) > \log(b)+1\right) &\leq& C_0e\frac{\pi_0}{1-\pi_0}\frac{n^{|Z|(|X|-1)+1}}{\log(b)+1} \\\nonumber
&=& C_0e\frac{n^{|Z|(|X|-1)+1}}{(\log(b)+1)n^{\lambda_0}},
\end{eqnarray}
and according to (\ref{eq:bound}),
\begin{eqnarray}\nonumber
\text{Pr}_{\text{shuffle}}\left(\text{BF}\left(X|Y,Z\right) > b\right) &\leq& \text{Pr}_{\text{shuffle}}\left( \max_{2 \leq s \leq t \leq n} \left[ ne\frac{\pi_0}{1-\pi_0} \Delta\left(s,t\right) \right] > \log(b)+1\right) \\\nonumber
&\leq& \sum_{2 \leq s \leq t \leq n} \text{Pr}_{\text{shuffle}}\left(ne\frac{\pi_0}{1-\pi_0} \Delta\left(s,t\right) > \log(b)+1\right) \\\nonumber
&\leq& C_1\frac{n^{|Z|(|X|-1)+3}}{(\log(b)+1)n^{\lambda_0}},
\end{eqnarray}
where $C_1 = C_0e/2$. 

Moreover, according to (\ref{eq:cover}) and $\text{Pr}_{H_0}\left(X|Z,Y\right) = \Psi\left(x_1,\ldots,x_n\right)$,
$$
\text{Pr}_{\text{shuffle}}\left( \left. \{x_i\}_{i=1}^n \right| \{z_i\}_{i=1}^n \right) \leq C_0n^{|Z|(|X|-1)}\text{Pr}_{H_0}\left(X|Z,Y\right),
$$
and
\begin{eqnarray}\nonumber
& & \text{Pr}_{\text{shuffle}}\left(\text{BF}\left(X|Y,Z\right) > b\right) \\\nonumber
& \leq & \sum_{\left(x_1,\ldots,x_n\right)} \frac{\text{BF}\left(X|Y,Z\right)}{b} \text{Pr}_{\text{shuffle}}\left( \left. \{x_i\}_{i=1}^n \right|
\{z_i\}_{i=1}^n \right) \\\nonumber
& \leq & \frac{C_0 n^{|Z|(|X|-1)}}{b} \sum_{\left(x_1,\ldots,x_n\right)}  \text{Pr}_{H_1}\left(X|Z,Y\right)\\\nonumber
&=& \frac{C_0n^{|Z|(|X|-1)}}{b} \leq \frac{C_1n^{|Z|(|X|-1)}}{b},
\end{eqnarray}
where $C_1 = C_0e/2 > C_0$, $\text{Pr}_{H_1}\left(X|Z,Y\right) $ is given by (\ref{eq:sum}) and $\sum_{\left(x_1,\ldots,x_n\right)}$ denotes the sum over all possible $|X|^{n}$ configurations of $\left(X_1,\ldots,X_n\right)$. The second inequality follows from (\ref{eq:bf}), and the last equality follows from the fact that  total probability is one. 
\end{proof}

\subsection{Proof of Corollary~\ref{cor:bf1}}\label{app:cor1}

\begin{proof}[Corollary~\ref{cor:bf1}]
We have the following inequalities on log-gamma functions:
$$
\log \Gamma(x) \leq (x-0.5)\log(x)-x + 1,
$$
and
$$
\log \Gamma(x) \geq (x-\gamma_0)\log(x)-x+1,
$$
where $x \geq 1$ and $\gamma_0=0.57722\ldots$ is the Euler-Mascheroni constant \citep{li2007}. Moreover,
$$
(x+c)\log(x+c) - x\log(x) \leq c\log(x) + c(1+c),
$$
for $x \geq 1$ and $c>0$. Therefore,
\begin{eqnarray}\label{eq:gamma}
& & \sum_{k=1}^{|X|}\log\Gamma\left(n_{j,k}+1\right)-\log\Gamma\left(n_{j}+|X|\right)\\\nonumber
&\geq& \sum_{k=1}^{|X|}n_{j,k}\log\left(n_{j,k}/n_j\right)-(|X|-1.5+\gamma_0)\log(n_j)-|X|(|X|+1),
\end{eqnarray}
for $j = 1, \ldots, |Z|$. Given $\alpha_0 \leq |X|$, the Dirichlet prior density of $p_j$, $f(p_j) \geq m_0$ for some $m_0>0$ and $j = 1, \ldots, |Z|$. When the sharp null hypothesis $H_0^{\text{sharp}}$ is true,
$$
\text{Pr}_{\text{sharp}}\left( \left. \{x_i\}_{i=1}^t \right| \{z_i\}_{i=1}^t \right) = \prod_{j=1}^{|Z|} \prod_{k=1}^{|X|} p_{j,k}^{n_{j,k}^{(1:t)}} \leq \prod_{j=1}^{|Z|} \prod_{k=1}^{|X|} \left(\frac{n_{j,k}^{(1:t)}}{n_j^{(1:t)}}\right)^{n_{j,k}^{(1:t)}},
$$
and
\begin{eqnarray}\nonumber
\Psi\left(x_1,\ldots,x_t\right) &\geq& m_0^{|Z|}\prod_{j=1}^{|Z|} \left[\frac{\prod_{k=1}^{|X|}\Gamma\left(n_{j,k}^{(1:t)}+1\right)}{\Gamma\left(n_j^{(1:t)}+|X|\right)}\right]\\\nonumber
&\geq& (1/B_0) n^{-|Z|(|X|-1.5+\gamma_0)}\prod_{j=1}^{|Z|} \prod_{k=1}^{|X|} \left(\frac{n_{j,k}^{(1:t)}}{n_j^{(1:t)}}\right)^{n_{j,k}^{(1:t)}},
\end{eqnarray}
where the last inequality follows from (\ref{eq:gamma}) with $B_0 = (e^{|X|(1+|X|)}/m_0)^{|Z|} > 0$. Therefore,
$$
\text{Pr}_{\text{sharp}}\left( \left. \{x_i\}_{i=1}^t \right| \{z_i\}_{i=1}^t \right) \leq B_0 n^{|Z|(|X|-1.5+\gamma_0)} \Psi\left(x_1,\ldots,x_t\right),
$$
and the rest of the proof follows the same argument as in the proof of Theorem~\ref{thm:bf1}.
\end{proof}

\subsection{Proof of Corollary~\ref{cor:bf2}}\label{app:cor2}

\begin{proof}[Corollary~\ref{cor:bf2}]
 Given $\alpha_0 \leq |X|$, the Dirichlet prior density of $p_j$, $f(p_j) \geq m_0$ for some $m_0>0$ and $j = 1, \ldots, |Z|$. When the hierarchical null hypothesis $H_0^{\text{hierar}}$ is true,
$$
\prod_{j=1}^{|Z|}f_j\left(p_j\right) \leq A_0 \prod_{j=1}^{|Z|}\Gamma\left(\alpha_0\right) \prod_{k=1}^{|X|}  \frac{p_{j,k}^{\frac{\alpha_0}{|X|}-1}}{\Gamma\left(\frac{\alpha_0}{|X|}\right)},
$$
where $f_j\left(p_j\right)\leq M_0$ for some $M_0 > 0$, $j=1,\ldots,|Z|$, and $A_0 = (M_0/m_0)^{|Z|}$. Thus,
\begin{eqnarray}\nonumber
& & \text{Pr}_{\text{hierar}}\left( \left. \{x_i\}_{i=1}^t \right| \{z_i\}_{i=1}^t \right) =\int\ldots\int \left[\prod_{j=1}^{|Z|} \prod_{k=1}^{|X|} p_{j,k}^{n_{j,k}^{(1:t)}}\right] \prod_{j=1}^{|Z|} f_j(p_j)dp_j \\\nonumber
& \leq & A_0 \int\ldots\int \left[\prod_{j=1}^{|Z|} \Gamma\left(\alpha_0\right) \prod_{k=1}^{|X|}  \frac{p_{j,k}^{n_{j,k}^{(1:t)}+\frac{\alpha_0}{|X|}-1}}{\Gamma\left(\frac{\alpha_0}{|X|}\right)} \right] \prod_{j=1}^{|Z|} dp_j \\\nonumber
&=& A_0\prod_{j=1}^{|Z|} \left[\frac{\Gamma(\alpha_0)}{\Gamma\left(\alpha_0 + n_j^{(1:t)}\right)}\prod_{k=1}^{|X|}\frac{\Gamma\left(n_{j,k}^{(1:t)}+\frac{\alpha_0}{|X|}\right)}{\Gamma\left(\frac{\alpha_0}{|X|}\right)}\right] 
= A_0\Psi\left(x_1,\ldots,x_t\right).
\end{eqnarray}
The rest of the proof follows the same argument as in the proof of Theorem~\ref{thm:bf1}.
\end{proof}

\subsection{Proof of Theorem~\ref{thm:bf2}}\label{app:thm2}

The alternative hypothesis $H_1$ is true if and only if the conditional mutual information between $X$ and $Y$ given $Z$, $\text{MI}(X,Y|Z)>0$. For any slicing scheme $S(\cdot)$, we define the plug-in estimator of the mutual information between $X$ and $S(Y)$ conditional on $Z$, $\text{MI}\left(X,S(Y)|Z\right)$, as
$$
\widehat{\text{MI}}\left(X,S(Y)|Z\right) = \frac{1}{n}\left[\sum_{j=1}^{|Z|}\sum_{h=1}^{|S|}\sum_{k=1}^{|X|}n_{j,k}^{(h)}\log\left(\frac{n_{j,k}^{(h)}}{n_j^{(h)}}\right) - \sum_{j=1}^{|Z|}\sum_{k=1}^{|X|}n_{j,k}\log\left(\frac{n_{j,k}}{n_j}\right)\right].
$$
The conditional mutual information $\text{MI}(X,Y|Z)$ is invariant to invertible transformations of $Y$ given $Z$, that is, $\text{MI}(X,Y|Z)=\text{MI}(X,U|Z)$ for $U = h_j (Y)$, where the function $h_j$ can be different for different $j$'s ($1 \leq j \leq |Z|$). Without loss of generality, we assume that $U|Z=j \sim \text{Unif}(0,1)$ for $1 \leq j \leq |Z|$  (since we can always apply the transformation $F(y|Z=j) = \text{Pr}\left(Y \leq y | Z=j\right)$ within each group $Z=j$ while keeping $\text{MI}(X,Y|Z)$ invariant). Let $f(u|Z=j)$ be the probability density function of $U$ given $Z=j$ and $f(u|X=k,Z=j)$ be the conditional probability density function of $U$ given $X=k$ and $Z=j$.  Since for $u \in [0,1]$,
$$
f(u|Z=j) = \sum_{k=1}^K f(u|X=k,Z=j) \text{Pr}\left(X=k|Z=j\right) = 1,
$$
there exists $B_0 > 0$ such that
$$
\max_{u \in [0,1], k \in \{1,\ldots,|X|\}, j \in \{1,\ldots,|Z|\}} \left[f(u|X=k,Z=j)\right]  \leq B_0.
$$ 
The following regularity condition guarantees that there exists a slicing scheme $\widetilde{S}_n(\cdot)$ such that $|\widetilde{S}_n|=O(|Z|n^{\gamma})$ ($\gamma<1$) and $\text{MI}(X,\widetilde{S}_n(Y)|Z) \rightarrow \text{MI}(X,Y|Z)$ as $n \rightarrow \infty$. It requires that the logarithm curve of the probability density does not have too many big ``jumps''. The condition is purely technical and only serve to provide theoretical understanding of the proposed method. We do not intend to make these assumptions the weakest possible.

\vspace{.1in}

\noindent \textbf{Regularity Condition:} Assume that the derivative $f'(u|X=k,Z=j)$ exists almost every where for $u \in [0,1]$ and $1 \leq k \leq |X|$, $1 \leq j \leq |Z|$. Let $\Lambda_{j,k} (\eta) = \{u \in [0,1]: |f'(u|X=k,Z=j)| \geq \eta f(u|X=k,Z=j) \}$ and $\Omega_{j,k} (\eta)$ denote collections of disjoint intervals in the set $\Lambda_{j,k} (\eta)$. We assume that there exist $\eta_0>0$ and $N_0 \geq 0$ such that for any $\eta \geq \eta_0$, $\sum_{j=1}^{|Z|}\sum_{k=1}^{|X|}|\Omega_{j,k}(\eta)| \leq N_0$, where $|\Omega_{j,k}(\eta)|$ is the number of disjoint intervals in $\Omega_{j,k}(\eta)$.

\vspace{.1in}

To prove Theorem~\ref{thm:bf2}, we will need the following lemma, the proof of which is given in the online supplement \citep{jiang2015}. 
\begin{lemma}\label{lem:bf}
Under the above regularity condition, there exists a slicing scheme $\widetilde{S}_n(\cdot)$ with  $|\widetilde{S}_n| = O(|Z|n^{2/3})$ such that for sufficiently large $n$,
$$
\mathrm{Pr}\left(\widehat{\mathrm{MI}}\left(X,\widetilde{S}_n(Y)|Z\right) \geq \mathrm{MI}\left(X,Y|Z\right)-\delta_0(n)\right) \geq  1-4n^{-\frac{1}{32}\log(n)}.
$$
where $\delta_0(n) = O\left(\frac{\log(n)}{n^{1/3}}\right) \rightarrow 0$ as $n \rightarrow \infty$.
\end{lemma}

\begin{proof}[Theorem~\ref{thm:bf2}]
First, according to Lemma~\ref{lem:bf}, there exists a slicing scheme $\widetilde{S}_n(\cdot)$ such that for sufficiently large $n$,
$$
\mathrm{Pr}\left(\widehat{\mathrm{MI}}\left(X,\widetilde{S}_n(Y)|Z\right) \geq \mathrm{MI}\left(X,Y|Z\right)-\delta_0(n)\right) \geq  1-4n^{-\frac{1}{32}\log(n)},
$$
where $|\widetilde{S}_n| = O(|Z|n^{2/3})$ and $\delta_0(n) = O\left(\frac{\log(n)}{n^{1/3}}\right)$.

Second, according to the log-gamma inequality (\ref{eq:gamma}) one can show that for $\alpha_0 \leq |X|$ and $|\widetilde{S}_n| \geq 2$, we have

\begin{eqnarray}\nonumber
& & \sum_{h=1}^{|\widetilde{S}_n|}\sum_{j=1}^{|Z|}\left[\sum_{k=1}^{|X|}\log\Gamma\left(n_{j,k}^{(h)}+\alpha_0/|X|\right)-\log\Gamma\left(n_{j}^{(h)}+\alpha_0\right)\right]\\\nonumber
& & - \sum_{j=1}^{|Z|}\left[\sum_{k=1}^{|X|}\log\Gamma\left(n_{j,k}+\alpha_0/|X|\right)-\log\Gamma\left(n_j+\alpha_0\right)\right]\\\nonumber
&\geq& \sum_{h=1}^{|\widetilde{S}_n|}\sum_{j=1}^{|Z|}\left[\sum_{k=1}^{|X|}\log\Gamma\left(n_{j,k}^{(h)}+1\right)-\log\Gamma\left(n_{j}^{(h)}+|X|\right)\right]\\\nonumber
& & -\sum_{j=1}^{|Z|}\left[\sum_{k=1}^{|X|}\log\Gamma\left(n_{j,k}\right)-\log\Gamma\left(n_j\right)\right]\\\nonumber
&\geq& \sum_{h=1}^{|\widetilde{S}_n|}\sum_{j=1}^{|Z|}\sum_{k=1}^{|X|}n_{j,k}^{(h)}\log\left(\frac{n_{j,k}^{(h)}}{n_j^{(h)}}\right)-\sum_{j=1}^{|Z|}\sum_{k=1}^{|X|}n_{j,k}\log\left(\frac{n_{j,k}}{n_j}\right)-c_0\\\nonumber
& & -(|X|-1.5+\gamma_0)\sum_{h=1}^{|\widetilde{S}_n|}\sum_{j=1}^{|Z|}\log(n_j^{(h)})-(\gamma_0-0.5)\sum_{j=1}^{|Z|}\log(n_j) \\\nonumber
&\geq& n\widehat{\mathrm{MI}}\left(X,\widetilde{S}_n(Y)|Z\right)-|Z|(|X|-1.5+\gamma_0)\log(n/|\widetilde{S}_n|)|\widetilde{S}_n|\\\nonumber
& & - (\gamma_0-0.5)|Z|\log(n)-c_0,
\end{eqnarray}
where $c_0$ is a constant that does not depend on $n$. So,
\begin{eqnarray}\nonumber
\text{BF}\left(X|Y,Z\right) &\geq& \text{BF}\left(X|\widetilde{S}_n(Y),Z\right) \pi_0^{|\widetilde{S}_n|-1}(1-\pi_0)^{n-|\widetilde{S}_n|} \\\nonumber
& \geq & \pi_0^{|\widetilde{S}_n|-1}(1-\pi_0)^{n-|\widetilde{S}_n|} e^{n \left(\widehat{\mathrm{MI}}\left(X,\widetilde{S}_n(Y)|Z\right) - \delta_1(n)\right)},
\end{eqnarray}
where 
\begin{eqnarray}\nonumber
\delta_1(n) &=& \frac{1}{n}\left[|Z|(|X|-1.5+\gamma_0)\log(n/|\widetilde{S}_n|)|\widetilde{S}_n| - (\gamma_0-0.5)|Z|\log(n) + c_0\right.\\\nonumber
& &\left. + |Z|\left(|X|\log\Gamma\left(\frac{\alpha_0}{|X|}\right)-\log\Gamma(\alpha_0)\right)(|\widetilde{S}_n|-1)\right]\\\nonumber
&=& O\left(\frac{|Z|(|X|-1.5+\gamma_0)\log(n/|\widetilde{S}_n|)|\widetilde{S}_n|}{n}\right)\\\nonumber
&=& O\left(\frac{|Z|^2(|X|-1.5+\gamma_0)\log(n)}{3n^{1/3}}\right).
\end{eqnarray}

Moreover, given  $\log\left(\pi_0/(1-\pi_0)\right) = -\lambda_0 \log(n)$ and $\lambda_0 \geq 1$, 
$$
(1-\pi_0)^{n-1} \left(\frac{\pi_0}{1-\pi_0}\right)^{(|S|-1)} \geq \frac{1}{e} \frac{1}{n^{\lambda_0(|S|-1)}},
$$
and
$$
\pi_0^{|\widetilde{S}_n|-1}(1-\pi_0)^{n-|\widetilde{S}_n|} \geq e^{-n\delta_2(n)},
$$
where 
$$
\delta_2(n) = O\left(\frac{\lambda_0\log(n)|\widetilde{S}_n|}{n}\right)= O\left(\frac{\lambda_0|Z|\log(n)}{n^{1/3}}\right).
$$

Therefore, 
$$
\mathrm{Pr}\left(\log\left[\text{BF}\left(X|Y,Z\right)\right] \geq n\left[\mathrm{MI}\left(X,Y|Z\right)-\delta(n)\right] \right) \geq  1-4n^{-\frac{1}{32}\log(n)},
$$
where $\delta(n)=\sum_{i=0}^2\delta_i(n)=O\left(\frac{(\lambda_0+|Z|(|X|-1.5+\gamma_0)/3)|Z|\log(n)}{n^{1/3}}\right) \rightarrow 0$ as $n \rightarrow \infty$. Thus,
$$
\text{BF}\left(X|Y,Z\right) \geq e^{n\left[\mathrm{MI}\left(X,Y|Z\right)-\delta(n)\right]}
$$
almost surely as $n \rightarrow \infty$.
\end{proof}

\section*{Acknowledgements}
This research was supported in part by  NSF grants DMS-1007762 and DMS-1120368. 

\begin{supplement}[id=suppA]
\stitle{Supplement to ``Bayesian Nonparametric Tests via Sliced Inverse Modeling''}
\slink[doi]{xx.xxx}
\sdatatype{.pdf}
\sdescription{We provide additional supporting materials that include detailed proofs and additional empirical results.}
\end{supplement}

\bibliographystyle{imsart-nameyear}
\bibliography{bf-reference}

\makeatletter
\renewcommand\@biblabel[1]{}
\makeatother

\end{document}